\documentclass[journal,12pt,onecolumn]{IEEEtran}
\usepackage[letterpaper, top=1in, bottom=1in, left=1in, right=1in]{geometry}

\usepackage{setspace}
\usepackage[table]{xcolor}
\usepackage{times}
\usepackage{epsfig}
\usepackage{amsmath}
\usepackage{amsthm}
\usepackage{amsfonts}
\usepackage{dsfont}
\usepackage{graphicx}
\usepackage{amssymb}
\usepackage{amstext}
\usepackage{latexsym}
\usepackage{color,colortbl}
\usepackage{ifthen}
\usepackage{multirow}
\usepackage{verbatim}
\usepackage{array,tabularx}
\usepackage{arydshln}
\usepackage[mathscr]{euscript}
\usepackage{accents}
 \usepackage{cite}
\usepackage{hhline}
\usepackage{caption}
\usepackage{subcaption}
\usepackage{enumerate}
\usepackage{xcolor}
\usepackage{mathtools}
\usepackage{url}
\usepackage{xparse}
\usepackage{makecell}
\usepackage{varwidth}
\usepackage{bm}
\usepackage{arydshln}
\usepackage{hyperref}

\usepackage{comment}
\usepackage{wrapfig}

\usepackage{caption}
\usepackage{float}
\usepackage{booktabs}

\usepackage{thmtools, thm-restate}

\usepackage{mathtools}

\usepackage{booktabs}

\usepackage{float}

\usepackage{multirow}

\newcolumntype{C}[1]{>{\centering\let\newline\\\arraybackslash\hspace{0pt}}m{#1}}

\usepackage[normalem]{ulem}

\usepackage{amsmath,pgfplots,amssymb}
\usepackage{subcaption,graphicx}

\newtheorem{theorem}{Theorem}
\newtheorem{theoremapp}{Theorem}

\newtheorem{lemma}{Lemma}

\newtheorem{proposition}{Proposition}
\newtheorem{propositionapp}{Proposition}

\theoremstyle{definition}
\newtheorem{example}{Example}
\newtheorem{definition}{Definition}

\theoremstyle{definition}
\newtheorem{remark}{Remark}

\theoremstyle{definition}

\DeclareMathOperator*{\argmax}{arg\,max}

\DeclareMathOperator*{\Inf}{Inf}

\usetikzlibrary{matrix,decorations.pathreplacing}
\usetikzlibrary{patterns}

\usepackage{pgfplots}
\usepgfplotslibrary{colormaps,fillbetween}
\usepackage{gnuplottex}
\usepgfplotslibrary{patchplots}

\definecolor{DarkGreen}{rgb}{0.1,0.5,0.1}
\definecolor{DarkRed}{rgb}{0.5,0.1,0.1}
\definecolor{DarkBlue}{rgb}{0.1,0.1,0.5}
\definecolor{DarkPurple}{rgb}{0.5,0.2,0.5}
\definecolor{DarkTurquoise}{rgb}{0.1,0.5,0.5}

\newcommand{\para}[1]{\textcolor{DarkRed}{ [#1 --PS] \normalsize }}

\newcommand{\sal}[1]{\textcolor{DarkBlue}{ [ #1 --salman ] \normalsize }}

\ifCLASSINFOpdf
\else
\fi
\begin{document}

%Title Suggestions

% \title{Three Choose Two from Low Influence, Utility and Independence for Differential Privacy}

\title{Low Influence, Utility, and Independence in Differential Privacy: A Curious Case of $3 \choose 2$}

%\title{3 Choose 2: Curious Dynamics between Low Influence, Utility and Independence in Differential Privacy}

% \title{Three Choose Two: Low Influence, Utility and Independence in Differential Privacy}

% \title{Low Influence, Utility and Independence in Differential Privacy: A Curious Case of Three Choose Two}


\author{Rafael G. L. D'Oliveira, Salman Salamatian, Muriel Médard, Parastoo Sadeghi   \\ RLE, Massachusetts Institute of Technology, USA\\ SEIT, University of New South Wales, Canberra, Australia \\ Emails: \{rafaeld, salmansa, medard\}@mit.edu, p.sadeghi@unsw.edu.au }

\maketitle

%\para{Tried a very different title, feel free to change!!}
\begin{abstract}
We study the relationship between randomized low influence functions and differentially private mechanisms. Our main aim is to formally determine whether differentially private mechanisms are low influence and whether low influence randomized functions can be differentially private. We show that differential privacy does not necessarily imply low influence in a formal sense. However, low influence implies \emph{approximate} differential privacy. These results hold for both independent and non-independent randomized mechanisms, where an important instance of the former is the widely-used additive noise techniques in the differential privacy literature. Our study also reveals the interesting dynamics between utility, low influence, and independence of a differentially private mechanism. As the name of this paper suggests, we show that any two such features are simultaneously possible. However, in order to have a differentially private mechanism that has both utility and low influence, even under a very mild utility condition, one has to employ non-independent mechanisms.
\end{abstract}

 %Specifically, an $\iota$-low influence randomized function implies $(\epsilon, \iota c)$-differential privacy for a constant $c$ and any $\epsilon \geq 0$.

%We study the relationship between low influence functions and differential privacy. We prove that independent mechanisms, a class of randomized mechanisms which contains the popular additive noise techniques from differential privacy, are not sufficient to simultaneously achieve low influence and differential privacy. We also show that differential privacy does not in fact imply low influence, even in general randomized schemes, thus formalizing the connection between stability of a randomized function and differential privacy. We characterize these conditions as geometric regions on the space of probabilities and  show that although functions can satisfy all possible combinations of being low influence or $\epsilon$-differentially private, the resulting mechanisms are no longer independent.
%This suggests that the study of non-independent mechanisms, which has been mostly overlooked in the literature, may be of interest when imposing additional constraints on the mechanisms.

% \fi
\IEEEpeerreviewmaketitle

\section{Introduction}

Since its inception in 2006, differential privacy \cite{Differential_privacy} has emerged as one of the main frameworks to design, evaluate, and implement privacy-preserving data analysis methods (see e.g., \cite{dwork2008differential,dwork2014algorithmic} for a survey of results). Some highlighted applications of differential privacy are Apple's large-scale private learning of users' preferences and behaviors \cite{apple}, and the 2020 United States Census' privatization method to provide data privacy protection \cite{uscensus}, each impacting hundreds of millions of individuals.

In the differential privacy framework, databases are mapped to randomized query outputs. The aim of such a randomized mechanism is to make the answer of a query (almost) statistically indistinguishable regardless of whether an individual participates in the database or not.
In other words, even if an adversary had knowledge of all records in the database before the participation of a particular individual, they would still have severely limited capability in inferring the private record of the individual from the query's output. In the differential privacy literature, this is often referred to as ensuring individuals in a database have a \emph{low influence} over the query's output, see for example \cite{dwork:rothblum:vadhan;focs:2010, kifer;ccs:2018}. While this intuitive understanding appears consistent with the definition of differential privacy, in this paper, we argue it is prone to ignoring the core objective of data sharing: to provide utility. It turns out that even taking a mild utility constraint into account, the influence (statistical power) of an individual over the traditional class of independent randomized response DP mechanisms can be quite substantial in an absolute sense.

In this paper, we aim to shed light on whether differentially private mechanisms are low influence mechanisms (or vice versa) in a more formal sense.

\subsection{Are differentially private mechanisms actually low influence?}

We begin by revising the definition of differential privacy, formalizing the notion of low influence, and then examining a toy example to establish our key findings.

\noindent \textbf{Setting:} We consider a set $\mathcal{D}$ of datasets with a neighborhood relationship. This relationship is a symmetric relationship on $\mathcal{D}$, denoted by $d \sim d'$ whenever $d,d' \in \mathcal{D}$ are neighbors. In the differential privacy literature, a neighborhood is oftentimes defined as two datasets that differ in only one entry (corresponding to the response from one individual). A mechanism is a randomized function $\mathcal{M}: \mathcal{D} \rightarrow \mathcal{V}$, where $\mathcal{V}$ is referred to as the output space of the mechanism.

\noindent \textbf{Differential Privacy and Low Influence Conditions:}
\begin{definition}\label{def: dp} A randomized mechanism $\mathcal{M}:\mathcal{D} \to \mathcal{V}$ is $(\epsilon,\delta)$-differentially private if for any $d \sim d'$, we have  $\Pr [\mathcal{M}(d) \in \mathcal{S}] \leq e^\epsilon \Pr [\mathcal{M}(d') \in \mathcal{S}] + \delta$, for every $S \subseteq \mathcal{V}$. 
\end{definition}
In the original definition \cite{Differential_privacy}, $\delta = 0$. The case for $\delta > 0$ is a common relaxation or approximation \cite{dwork:mironov:delta:2006,Differential_privacy}. Usually, $\delta$ much smaller than $1$ is desired. Throughout the paper, we will use the shorthand  $(\epsilon,\delta)$-DP instead of both $(\epsilon,\delta)$-differential privacy or $(\epsilon,\delta)$-differentially private.

%\noindent \textbf{Low Influence Conditions:} 
We propose the following adaptation of the notion of low influence from \cite{penrose1946elementary}.\footnote{This notion is widely used in the fields of social choice theory and in the analysis of Boolean functions \cite{o2014analysis}.} 
\begin{definition}\label{def: low influence mechanism}
A randomized mechanism $\mathcal{M}:\mathcal{D} \to \mathcal{V}$ is $\iota$-low influence ($\iota$-LI for short) if for any $d \sim d'$, we have $\Pr[\mathcal{M}(d) \neq \mathcal{M}(d')] \leq \iota$. 
\end{definition}
In other words, the output of a randomized mechanism does not change much, statistically speaking, for neighboring datasets. Indeed, if $\iota = 0$, then no two neighboring datasets can be statistically distinguished. 

In short, the DP constraints on the probabilities are relative (modulo a constant factor $\delta$), while the LI constraints are absolute. To see why such a definition matters and how it is different from the DP constraints, consider a dataset containing YES/NO private votes from $2n+1$ individuals on a sensitive matter. Assume an adversary somehow knows $n$ individuals voted YES and another $n$ individuals voted NO. Hence, the majority is decided by the last individual, whose vote is unknown to the adversary. Assume a mildly useful (truthful) DP mechanism to approximate the majority function, where a 55\% biased coin is tossed to determine whether to reveal the true majority outcome. In the DP framework, the deciding vote is said to be statistically indistinguishable by a rather small multiplicative factor $\epsilon =\ln(\frac{0.55}{0.45}) \approx 0.2$. However, the probability of giving the \emph{same} YES or NO answer, even when the individual voted differently, is only $0.55\times 0.45\times 2 = 0.495$. In other words, $\iota = 0.505$ and more than half of the time the deciding vote influences the randomized outcome. In such randomized coin-toss models, insisting on even the slightest utility (slightest bias towards telling the truth) results in $\iota \geq 0.5$ -- a 50\% or more chance for an individual to influence the majority outcome in a borderline dataset.

It may appear that the only way to achieve lower than 50\% influence is through being statistically not truthful in one realization of the majority. For example, assume the mechanism declares YES as majority 91\% and 89\% of the time when the actual majority is YES and NO, respectively. At the cost of rendering the mechanism almost useless (giving the wrong answer with probability $0.5\times(0.09+0.89) = 0.49$ assuming equally-likely votes), the probability of giving the same answer is raised to $0.91\times 0.89+0.09\times 0.11 = 0.8198$, meaning $\iota = 0.1802$, while still achieving $\epsilon =\max\{\ln(\frac{0.11}{0.09}),\ln(\frac{0.91}{0.89}) \} \approx 0.2$ according to the DP framework.

 We show in this paper that there is a systematic way to address the above issue through generalizing the randomized response model from independent across datasets to \emph{joint} across datasets. To see how this can be fixed in the above scenario, refer to Example \ref{ex: nonindep mech}.

\noindent \textbf{On Mechanism Independence:} In studying the low influence conditions it will be important whether the mechanism is independent or not. A mechanism is independent if the random variables $\{ \mathcal{M}(d_i) : d_i \in \mathcal{D} \}$ are mutually independent. 

To the best of the authors' knowledge, the differential privacy framework has been solely based on independent mechanisms. Indeed, one of the most common ways of constructing differentially private schemes for continuous-valued queries is to add a noise variable $N_u$, such as Laplace  or Gaussian noise \cite{Differential_privacy,dwork2014algorithmic} to the true query output $f(d) = u$ to obtain $\mathcal{M}(d) = u + N_u$. For discrete queries taking values over a finite field $\mathbb F_q$, this construction is equivalent to simulating noise over a discrete memoryless channel with transition probability $\Pr[\mathcal{M}(d) = v|f(d) = u]$. 

Even if the noise statistics are query-output dependent, as is proposed in \cite{nissim:smith:query:dependent:2007} (see also \cite{Viswanath_DP_Staircase} for more discussion), such schemes are still independent mechanisms. This is because the random variables $\{N_u\}$ are mutually independent and \[\Pr[\mathcal{M}(d_i) = v_i, \cdots, \mathcal{M}(d_j) = v_j|f(d_i) = u_i, \cdots, f(d_j) = u_j] \] decomposes into the product $\prod_{\ell = i}^j \Pr[\mathcal{M}(d_\ell) = v_\ell|f(d_\ell) = u_\ell]$ for any $i,j \in \{1,\cdots, |\mathcal D|\}$, $i\leq j$. In this paper, we will consider both independent and general (i.e., non-independent) mechanisms on $\mathcal D$. By a non-independent mechanism, we mean that the random variables $\{ \mathcal{M}(d_i) : d_i \in \mathcal{D} \}$ are in general  dependent upon each other.

As alluded to previously, at the surface level, the low influence and differential privacy conditions may seem synonymous to each other. Our first key observation is that this is not the case. As it turns out, the low influence and differential privacy conditions do not generally imply one another. However, low influence does imply the approximate form of differential privacy, i.e., for $\delta>0$. We illustrate this in the following example.

\begin{figure*}[!t]
\begin{subfigure}[t]{0.30\textwidth}
    \centering
    \begin{tikzpicture}
\begin{axis}[width=5cm, height=5cm,
stack plots=y,thick,smooth,no markers,xmin=0,xmax=1,ymin=0,ymax=1,xlabel={$x$},
  ylabel={$y$},axis on top,xtick={0,0.5,1},ytick={0,0.5,1},ylabel near ticks,xlabel near ticks]

%uncomment this to view grid:
% \draw[step=0.5cm,black,thin,dashed] (0,0) grid (100,100);

\filldraw[fill=blue!30, draw=blue!30] (0,0) -- (33.33,66.66) -- (100,100) -- (66.66,33.33) -- cycle;

\end{axis}
\end{tikzpicture}
    \caption{$\log(2)$-DP}
    \label{fig: nonlinear LI implies DP a}
\end{subfigure}
\begin{subfigure}[t]{0.30\textwidth}
    \centering
    \begin{tikzpicture}
\begin{axis}[width=5cm, height=5cm,
stack plots=y,thick,smooth,no markers,xmin=0,xmax=1,ymin=0,ymax=1,xlabel={$x$},
  ylabel={$y$},axis on top,xtick={0,0.5,1},ytick={0,0.5,1},ylabel near ticks,xlabel near ticks]

%uncomment this to view grid:
% \draw[step=0.5cm,black,thin,dashed] (0,0) grid (100,100);

\filldraw[fill=blue!30, draw=blue!30] (0,0) -- (0,40) -- (60,100) -- (100,100) -- (100,60) -- (40,0) -- cycle;

\end{axis}
\end{tikzpicture}
    \caption{$(0,\frac{2}{5})$-DP}
    \label{fig: nonlinear LI implies DP b}
\end{subfigure}
\begin{subfigure}[t]{0.30\textwidth}
    \centering
    \begin{tikzpicture}
%uncomment this to view grid:
%\draw[step=0.5cm,black,thin,dashed] (0,0) grid (10,10);

\begin{axis}[width=5cm, height=5cm,
stack plots=y,thick,smooth,no markers,xmin=0,xmax=1,ymin=0,ymax=1,xlabel={$x$},
  ylabel={$y$},axis on top,xtick={0,0.5,1},ytick={0,0.5,1}]
  
\draw[scale=100,fill=blue!30,color=blue!30,domain=0:0.498,variable=\x] plot ({\x},{(5*\x-2)/(10*\x-5)});

\draw[scale=100,fill=blue!30,color=blue!30,domain=0.503:1,variable=\x] plot ({\x},{(5*\x-2)/(10*\x-5)});

\filldraw[fill=blue!30, draw=blue!30] (0,0) -- (0,40) -- (20,0)-- cycle;

\filldraw[fill=blue!30, draw=blue!30] (100,100) -- (100,60) -- (80,100)-- cycle;
\end{axis}
\end{tikzpicture}
    \caption{$\frac{2}{5}$-LI}
    \label{fig: nonlinear LI implies DP c}
\end{subfigure}
\caption{Figure (a) and (b) illustrate the $(\epsilon,\delta)$-DP regions, and Figure (c) shows the $\iota$-LI region for independent mechanisms, where $x = \Pr[\mathcal{M}(d_1)=1]$ and $y = \Pr[\mathcal{M}(d_2)=1]$.}
\label{fig: nonlinear LI implies DP}
\end{figure*}
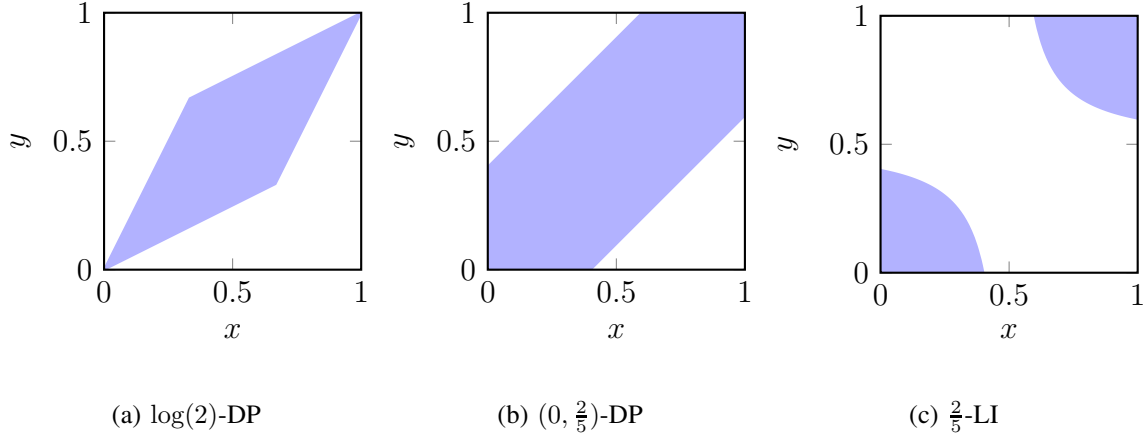 

\begin{example}\label{toy1} Consider the case where $\mathcal{D}$ consists of two neighboring datasets ${d_1 \sim d_2}$, and the output space is binary, i.e. $\mathcal{V} = \{1,2 \}$. For the sake of brevity, let $x = \Pr[\mathcal{M}(d_1) = 1]$ and $y = \Pr[\mathcal{M}(d_2) = 1]$. Thus, the full set of $(\epsilon,\delta)$-DP conditions are given by
    \begin{align*}
    x\leq e^{\epsilon} y + \delta , \quad y\leq e^{\epsilon} x + \delta , \quad 1-x\leq e^{\epsilon} (1-y) + \delta , \quad \text{and} \quad 1-y\leq e^{\epsilon} (1-x) + \delta .
\end{align*}
In Figs.~\ref{fig: nonlinear LI implies DP a} and \ref{fig: nonlinear LI implies DP b}, we show the set of points $(x,y) \in \mathbb{R}^2$ that satisfy these constraints, for $\epsilon = \log(2),\delta = 0$ and for $\epsilon = 0,\delta = \frac{2}{5}$, respectively. 

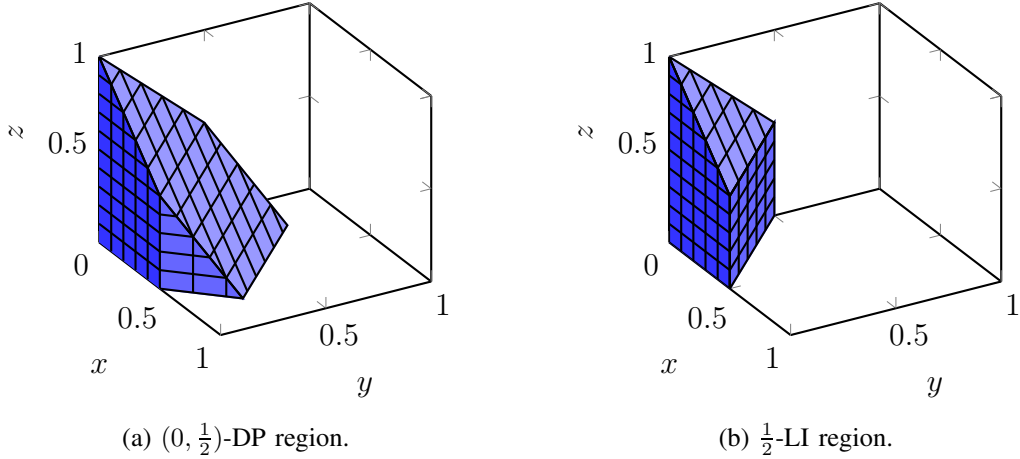
\begin{figure*}[!t]
\begin{subfigure}[]{0.45\textwidth}
    \centering
    \begin{tikzpicture}

\begin{axis}[width=6cm, height=6cm,stack plots=y,thick,smooth,no markers,xmin=0,xmax=1,ymin=0,ymax=1,zmin=0,zmax=1,xlabel={$x$},
  ylabel={$y$},zlabel={$z$},xtick={0,0.5,1},ytick={0.5,1},ztick={0.5,1},view={60}{30}]

\draw[color=black] (0,0,10) -- (50,0,10);
\draw[color=black] (0,0,20) -- (50,0,20);
\draw[color=black] (0,0,30) -- (50,0,30);
\draw[color=black] (0,0,40) -- (50,0,40);
\draw[color=black] (0,0,50) -- (50,0,50);
\draw[color=black] (0,0,60) -- (40,0,60);
\draw[color=black] (0,0,70) -- (30,0,70);
\draw[color=black] (0,0,80) -- (20,0,80);
\draw[color=black] (0,0,90) -- (10,0,90);

\draw[color=black] (10,0,0) -- (10,0,90);
\draw[color=black] (20,0,0) -- (20,0,80);
\draw[color=black] (30,0,0) -- (30,0,70);
\draw[color=black] (40,0,0) -- (40,0,60);
\draw[color=black] (50,0,0) -- (50,0,50);

\draw[color=black] (50,0,0) -- (75,25,0);
\draw[color=black] (50,0,10) -- (70,20,10);
\draw[color=black] (50,0,20) -- (65,15,20);
\draw[color=black] (50,0,30) -- (60,10,30);
\draw[color=black] (50,0,40) -- (55,5,40);

\draw[color=black] (50,0,50) -- (75,25,0);

\draw[color=black] (60,10,30) -- (60,10,0);
\draw[color=black] (70,20,10) -- (70,20,0);

\draw[color=black] (0,0,100) -- (50,0,50);
\draw[color=black] (0,10,90) -- (60,10,30);
\draw[color=black] (0,20,80) -- (70,20,10);
\draw[color=black] (0,30,70) -- (70,30,0);
\draw[color=black] (0,40,60) -- (60,40,0);
\draw[color=black] (0,50,50) -- (50,50,0);
\draw[color=black] (10,60,30) -- (40,60,0);
\draw[color=black] (20,70,10) -- (30,70,0);

\draw[color=black] (10,0,90) -- (0,10,90);
\draw[color=black] (20,0,80) -- (0,20,80);
\draw[color=black] (30,0,70) -- (0,30,70);
\draw[color=black] (40,0,60) -- (0,40,60);
\draw[color=black] (50,0,50) -- (0,50,50);

\draw[color=black] (55,5,40) -- (5,55,40);
\draw[color=black] (60,10,30) -- (10,60,30);
\draw[color=black] (65,15,20) -- (15,65,20);
\draw[color=black] (70,20,10) -- (20,70,10);

\draw[color=black] (75,25,0) -- (50,0,50) -- (0,0,100) -- (0,50,50) -- (25,75,0) -- cycle;

\filldraw[color=blue!80] (0,0,0) -- (0,0,100) -- (50,0,50) -- (50,0,0) -- cycle;

\filldraw[color=blue!60] (50,0,0) -- (50,0,50) -- (75,25,0) -- cycle;

\filldraw[color=blue!40] (75,25,0) -- (50,0,50) -- (0,0,100) -- (0,50,50) -- (25,75,0) -- cycle;

\end{axis}

\end{tikzpicture}
    \caption{$(0,\frac{1}{2})$-DP region.}
    \label{fig: LI implies DP a}
\end{subfigure}
\begin{subfigure}[]{0.45\textwidth}
    \centering
    \begin{tikzpicture}

\begin{axis}[width=6cm, height=6cm,stack plots=y,thick,smooth,no markers,xmin=0,xmax=1,ymin=0,ymax=1,zmin=0,zmax=1,xlabel={$x$},
  ylabel={$y$},zlabel={$z$},xtick={0,0.5,1},ytick={0.5,1},ztick={0.5,1},view={60}{30}]

\draw[color=black] (0,0,10) -- (50,0,10);
\draw[color=black] (0,0,20) -- (50,0,20);
\draw[color=black] (0,0,30) -- (50,0,30);
\draw[color=black] (0,0,40) -- (50,0,40);
\draw[color=black] (0,0,50) -- (50,0,50);
\draw[color=black] (0,0,60) -- (40,0,60);
\draw[color=black] (0,0,70) -- (30,0,70);
\draw[color=black] (0,0,80) -- (20,0,80);
\draw[color=black] (0,0,90) -- (10,0,90);

\draw[color=black] (10,0,0) -- (10,0,90);
\draw[color=black] (20,0,0) -- (20,0,80);
\draw[color=black] (30,0,0) -- (30,0,70);
\draw[color=black] (40,0,0) -- (40,0,60);
\draw[color=black] (50,0,0) -- (50,0,50);

\draw[color=black] (0,0,100) -- (50,0,50);
\draw[color=black] (0,10,90) -- (40,10,50);
\draw[color=black] (0,20,80) -- (30,20,50);
\draw[color=black] (0,30,70) -- (20,30,50);
\draw[color=black] (0,40,60) -- (10,40,50);

\draw[color=black] (10,0,90) -- (0,10,90);
\draw[color=black] (20,0,80) -- (0,20,80);
\draw[color=black] (30,0,70) -- (0,30,70);
\draw[color=black] (40,0,60) -- (0,40,60);
\draw[color=black] (50,0,50) -- (0,50,50);

\draw[color=black] (50,0,50) -- (50,0,0);
\draw[color=black] (40,10,50) -- (40,10,0);
\draw[color=black] (30,20,50) -- (30,20,0);
\draw[color=black] (20,30,50) -- (20,30,0);
\draw[color=black] (10,40,50) -- (10,40,0);

\draw[color=black] (50,0,40) -- (0,50,40);
\draw[color=black] (50,0,30) -- (0,50,30);
\draw[color=black] (50,0,20) -- (0,50,20);
\draw[color=black] (50,0,10) -- (0,50,10);

\draw[color=black] (0,0,0) -- (0,0,100) -- (50,0,50) -- (50,0,0) -- cycle;

\draw[color=black] (50,0,0) -- (50,0,50) -- (0,50,50) -- (0,50,0) -- cycle;

\draw[color=black] (0,0,100) -- (50,0,50) -- (0,50,50) -- cycle;

\filldraw[color=blue!80] (0,0,0) -- (0,0,100) -- (50,0,50) -- (50,0,0) -- cycle;

\filldraw[color=blue!60] (50,0,0) -- (50,0,50) -- (0,50,50) -- (0,50,0) -- cycle;

\filldraw[color=blue!40] (0,0,100) -- (50,0,50) -- (0,50,50) -- cycle;

\end{axis}

\end{tikzpicture}
    \caption{$\frac{1}{2}$-LI region.}
    \label{fig: LI implies DP b}
\end{subfigure}
\caption{$(\epsilon,\delta)$-DP and $\iota$-LI regions for non-independent mechanisms for $\delta = \iota =1/2$.}
\label{fig: LI implies DP}
\end{figure*}
 
On the other hand, an independent mechanism is $\iota$-LI if and only if $x$ and $y$ satisfy the following hyperbolic constraint.
\begin{align*}
    \iota &\geq \Pr[\mathcal{M}(d_1) \neq \mathcal{M}(d_2)] = 1 - \Pr[\mathcal{M}(d_1) = \mathcal{M}(d_2)] \\
    &= 1 - \Pr[\mathcal{M}(d_1) = 1 , \mathcal{M}(d_2) = 1] - \Pr[\mathcal{M}(d_1) = 2 , \mathcal{M}(d_2) = 2] \\
    &= 1 - \Pr[\mathcal{M}(d_1) = 1 ] \Pr[\mathcal{M}(d_2) = 1] - \Pr[\mathcal{M}(d_1) = 2 ] \Pr[\mathcal{M}(d_2) = 2] \\
    &= x+y-2xy.
\end{align*}
In Fig.~\ref{fig: nonlinear LI implies DP c}, we show the set of points $(x,y) \in \mathbb{R}^2$ that satisfy this constraint for $\iota = \frac{2}{5}$. 

Fig.~\ref{fig: nonlinear LI implies DP} illustrates how, when $\delta = 0$, the $(\epsilon,\delta)$-DP region and the $\iota$-LI region are distinct and neither can be embedded into the other.  However, an $\iota$-LI region can be embedded into an $(\epsilon,\iota)$-DP region for any $\epsilon \geq 0$, meaning that $\iota$-LI implies $(\epsilon,\iota)$-DP for independent mechanisms. 

This is also true for non-independent mechanisms. Indeed, let ${x = \Pr [\mathcal{M}(d_1)=1 , \mathcal{M}(d_2)=2]}$, ${y = \Pr [\mathcal{M}(d_1)=2 , \mathcal{M}(d_2)=1]}$, and ${z = \Pr [\mathcal{M}(d_1)=1 , \mathcal{M}(d_2)=1]}$, subject to being inside the probability simplex: $0 \leq x+y+z \leq 1$. To ensure $\iota$-LI, one must have 
\begin{align*}
    \iota &\geq \Pr[\mathcal{M}(d_1) \neq \mathcal{M}(d_2)] = 1 - \Pr[\mathcal{M}(d_1) = \mathcal{M}(d_2)] \\
    &= 1 - \Pr[\mathcal{M}(d_1) = 1 , \mathcal{M}(d_2) = 1] - \Pr[\mathcal{M}(d_1) = 2 , \mathcal{M}(d_2) = 2] \\
    &= 1-z - (1-x-y-z) = x+y.
\end{align*}
For brevity, we detail only two of the $(\epsilon,\delta)$-DP conditions below
\begin{align*}
    x+z\leq e^{\epsilon} (y+z) + \delta \quad \text{and} \quad 1-x-z\leq e^{\epsilon} (1-y-z) + \delta.
\end{align*}

Fig.~\ref{fig: LI implies DP} shows the set of points $(x,y,z) \in \mathbb{R}^3$, which satisfy $(0,\frac{1}{2})$-DP and $\frac{1}{2}$-LI conditions (and also belong to the probability simplex). We can see from the figure that the low influence region, Fig.~\ref{fig: LI implies DP b}, is contained in the differential privacy region, Fig.~\ref{fig: LI implies DP a}.

\begin{figure*}[!t]
\begin{subfigure}[]{0.45\textwidth}
    \centering
    \begin{tikzpicture}
%uncomment this to view grid:
%\draw[step=0.5cm,black,thin,dashed] (0,0) grid (10,10);

\begin{axis}[width=6cm, height=6cm,
stack plots=y,thick,smooth,no markers,xmin=0,xmax=1,ymin=0,ymax=1,xlabel={$x$},
  ylabel={$y$},axis on top,xtick={0,0.5,1},ytick={0,0.5,1}]
  
\draw[scale=100,fill=blue!30,color=blue!30,domain=0:0.498,variable=\x] plot ({\x},{(5*\x-2)/(10*\x-5)});

\draw[scale=100,fill=blue!30,color=blue!30,domain=0.503:1,variable=\x] plot ({\x},{(5*\x-2)/(10*\x-5)});

\filldraw[fill=blue!30, draw=blue!30] (0,0) -- (0,40) -- (20,0)-- cycle;

\filldraw[fill=blue!30, draw=blue!30] (100,100) -- (100,60) -- (80,100)-- cycle;
\end{axis}
\end{tikzpicture}
    \caption{$\frac{2}{5}$-LI region}
    \label{fig: independent low influence is trivial a}
\end{subfigure}
\begin{subfigure}[]{0.45\textwidth}
    \centering
    \begin{tikzpicture}
        \begin{axis}[
            width=6cm, height=6cm,
            thick,smooth,no markers,
            xlabel={$x$},
            ylabel={$y$},
            xmin = 0, xmax = 1,
            ymin = 0, ymax = 1,
            zmin = 0, zmax = 1,
            axis equal image,
            view = {0}{90},
            samples  = 5,
            xtick={0,0.5,1},ytick={0,0.5,1},
            axis on top,
        ]

        \filldraw[color=red!30] (50,0) -- (50,50) -- (100,50) -- (100,0) -- cycle;
        
        \draw[color=blue!50, line width = 2pt] (50,0) -- (50,50) -- (100,50);

        \node[color=blue] at (50,55) {$S$};
        
        \node[color=red] at (80,30) {$R$};
        
            \addplot3[
                quiver = {
                    u = {(2*y-1)/((2*y-1)^2 + (2*x-1)^2)^(1/2)},
                    v = {(2*x-1)/((2*y-1)^2 + (2*x-1)^2)^(1/2)},
                    scale arrows = 0.05,
                },
                -stealth,
                domain = 0.55:1,
                domain y = 0:0.45,
                thick,
            ] {0};
            \end{axis}
    \end{tikzpicture}
    \caption{Vector field for nontrivial mechanisms.}
    \label{fig: independent low influence is trivial b}
\end{subfigure}
\caption{Figure (a) illustrates how an independent mechanism can only have arbitrarily low influence if it is trivial, i.e. $x$ and $y$ are both close to $0$ or close to $1$. Figure (b) shows that the influence for a nontrivial independent mechanism, i.e. belonging to the red region $R$, is lower bounded by the influence in the blue boundary region $S$, denoted by $I(S)$. We show that $I(S)=\frac{1}{2}$.}
\label{fig: independent low influence is trivial}
\end{figure*}
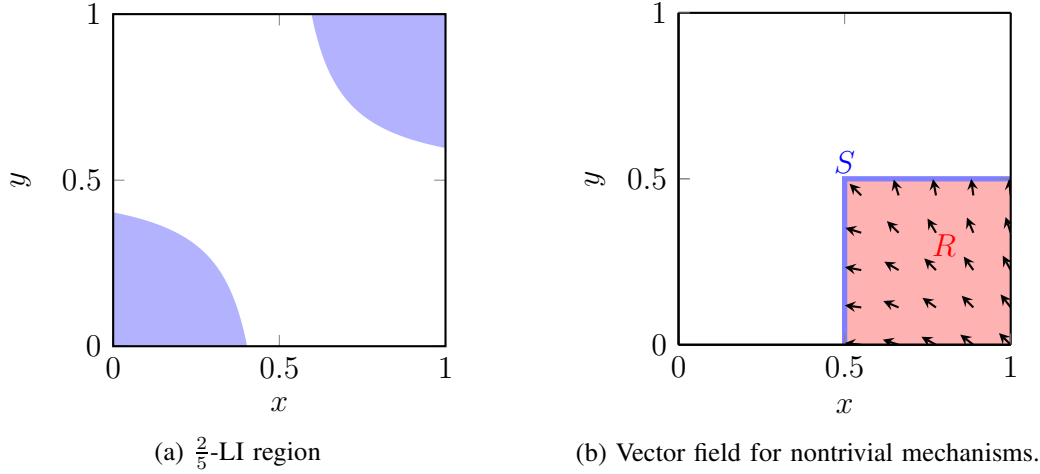
\end{example}
\subsection{What about the relationship between low influence and utility?}

In application, randomized mechanisms are generally used to approximate some function of interest. This approximation is often measured through some notion of utility, which is highly dependent on the application. It is then natural to ask if $\iota$-LI mechanisms can be useful, i.e. can a mechanism have utility while being low influence. It turns out that the answer to this heavily depends on whether the mechanism is independent, operating on single datasets, or not, jointly operating on the space of datasets. To formalize this, we introduce a universal notion of utility without being too restrictive. Roughly speaking, we say a mechanism is \emph{nontrivial} if the most likely output of the mechanism shows some variability across the datasets. Let us continue with the toy example.

\textbf{Example \ref{toy1} Continued:} We say the binary mechanism operating on $d_1 \sim d_2$ is nontrivial if outputting $v=0$ is more likely when the database is $d_1$ and outputting $v=1$ is more likely when the database is $d_2$ (or vice versa).\footnote{If the most likely query output is $v=0$ (or $v=1$) regardless of $d_1$ or $d_2$, then there will be no need to query and a fixed guess will always give the maximum likelihood outcome. Recall the voting example, where in two marginal neighboring datasets, being nontrivial means the randomized response must give the true majority answer with a probability more than 0.5.}

Looking again at independent mechanisms, and setting $x = \Pr[\mathcal{M}(d_1) = 1]$ and $y = \Pr[\mathcal{M}(d_2) = 1]$, a mechanism is nontrivial if $x \geq 1-x$ and $y \leq 1-y$ (or vice versa). Without loss of generality, let us assume that $x \geq \frac{1}{2}$ and $y\leq \frac{1}{2}$.\footnote{To satisfy the nontriviality condition, we additionally assume $x$ and $y$ cannot simultaneously be equal to $\frac{1}{2}$.} It turns out that minimizing the influence function of the independent mechanism $I(x,y)= x+y-2xy$ subject to $x \geq \frac{1}{2}$ and $y\leq \frac{1}{2}$ leads to $\min I(x,y) = \frac{1}{2}$. Thus, if an independent mechanism is $\iota$-LI, then $\iota \geq \frac{1}{2}$. This result is illustrated in Figure \ref{fig: independent low influence is trivial}. In Section \ref{sec: nontrivial independent not LI}, we show that this is also true for non-binary mechanisms. \emph{Thus, nontrivial independent mechanisms cannot be arbitrarily low influence.}

However, if we consider general mechanisms, we can achieve both low influence $\iota < \frac{1}{2}$ and a nontrivial utility. Let ${x = \Pr [\mathcal{M}(d_1)=1 , \mathcal{M}(d_2)=2]}$, ${y = \Pr [\mathcal{M}(d_1)=2 , \mathcal{M}(d_2)=1]}$, and ${z = \Pr [\mathcal{M}(d_1)=1 , \mathcal{M}(d_2)=1]}$, subject to the probability simplex: $0 \leq x+y+z \leq 1$. Fig.~\ref{fig: LI can be nontrivial} illustrates how the nontrivial region specified by $x+z >\frac{1}{2}$ and $y+z < \frac{1}{2}$ (Fig.~ \ref{fig: LI can be nontrivial b}) intersects with the $\frac{1}{2}$-LI region specified by $x+y \leq \frac{1}{2}$ (Fig.~\ref{fig: LI can be nontrivial a}). This holds true in general, i.e., for any influence $\iota > 0$.

\begin{figure*}[!t]
\begin{subfigure}[]{0.45\textwidth}
    \centering
    \begin{tikzpicture}

\begin{axis}[width=6cm, height=6cm,stack plots=y,thick,smooth,no markers,xmin=0,xmax=1,ymin=0,ymax=1,zmin=0,zmax=1,xlabel={$x$},
  ylabel={$y$},zlabel={$z$},xtick={0,0.5,1},ytick={0.5,1},ztick={0.5,1},view={70}{30}]

\draw[color=black] (0,0,10) -- (50,0,10);
\draw[color=black] (0,0,20) -- (50,0,20);
\draw[color=black] (0,0,30) -- (50,0,30);
\draw[color=black] (0,0,40) -- (50,0,40);
\draw[color=black] (0,0,50) -- (50,0,50);
\draw[color=black] (0,0,60) -- (40,0,60);
\draw[color=black] (0,0,70) -- (30,0,70);
\draw[color=black] (0,0,80) -- (20,0,80);
\draw[color=black] (0,0,90) -- (10,0,90);

\draw[color=black] (10,0,0) -- (10,0,90);
\draw[color=black] (20,0,0) -- (20,0,80);
\draw[color=black] (30,0,0) -- (30,0,70);
\draw[color=black] (40,0,0) -- (40,0,60);
\draw[color=black] (50,0,0) -- (50,0,50);

\draw[color=black] (0,0,100) -- (50,0,50);
\draw[color=black] (0,10,90) -- (40,10,50);
\draw[color=black] (0,20,80) -- (30,20,50);
\draw[color=black] (0,30,70) -- (20,30,50);
\draw[color=black] (0,40,60) -- (10,40,50);

\draw[color=black] (10,0,90) -- (0,10,90);
\draw[color=black] (20,0,80) -- (0,20,80);
\draw[color=black] (30,0,70) -- (0,30,70);
\draw[color=black] (40,0,60) -- (0,40,60);
\draw[color=black] (50,0,50) -- (0,50,50);

\draw[color=black] (50,0,50) -- (50,0,0);
\draw[color=black] (40,10,50) -- (40,10,0);
\draw[color=black] (30,20,50) -- (30,20,0);
\draw[color=black] (20,30,50) -- (20,30,0);
\draw[color=black] (10,40,50) -- (10,40,0);

\draw[color=black] (50,0,40) -- (0,50,40);
\draw[color=black] (50,0,30) -- (0,50,30);
\draw[color=black] (50,0,20) -- (0,50,20);
\draw[color=black] (50,0,10) -- (0,50,10);

\draw[color=black] (0,0,0) -- (0,0,100) -- (50,0,50) -- (50,0,0) -- cycle;

\draw[color=black] (50,0,0) -- (50,0,50) -- (0,50,50) -- (0,50,0) -- cycle;

\draw[color=black] (0,0,100) -- (50,0,50) -- (0,50,50) -- cycle;

\filldraw[color=blue!80] (0,0,0) -- (0,0,100) -- (50,0,50) -- (50,0,0) -- cycle;

\filldraw[color=blue!60] (50,0,0) -- (50,0,50) -- (0,50,50) -- (0,50,0) -- cycle;

\filldraw[color=blue!40] (0,0,100) -- (50,0,50) -- (0,50,50) -- cycle;

\end{axis}

\end{tikzpicture}
    \caption{$\frac{1}{2}$-LI region.}
    \label{fig: LI can be nontrivial a}
\end{subfigure}
\begin{subfigure}[]{0.45\textwidth}
    \centering
    \begin{tikzpicture}

\begin{axis}[width=6cm, height=6cm,stack plots=y,thick,smooth,no markers,xmin=0,xmax=1,ymin=0,ymax=1,zmin=0,zmax=1,xlabel={$x$},
  ylabel={$y$},zlabel={$z$},xtick={0,0.5,1},ytick={0.5,1},ztick={0.5,1},view={70}{30}]

\draw[color=black] (0,0,50) -- (50,0,50);
\draw[color=black] (10,0,40) -- (60,0,40);
\draw[color=black] (20,0,30) -- (70,0,30);
\draw[color=black] (30,0,20) -- (80,0,20);
\draw[color=black] (40,0,10) -- (90,0,10);
\draw[color=black] (50,0,0) -- (100,0,0);

\draw[color=black] (10,0,50) -- (10,0,40);
\draw[color=black] (20,0,50) -- (20,0,30);
\draw[color=black] (30,0,50) -- (30,0,20);
\draw[color=black] (40,0,50) -- (40,0,10);
\draw[color=black] (50,0,50) -- (50,0,0);
\draw[color=black] (60,0,40) -- (60,0,0);
\draw[color=black] (70,0,30) -- (70,0,0);
\draw[color=black] (80,0,20) -- (80,0,0);
\draw[color=black] (90,0,10) -- (90,0,0);

\draw[color=black] (50,0,50) -- (100,50,0);

\draw[color=black] (50,0,50) -- (100,0,0);
\draw[color=black] (60,10,40) -- (100,10,0);
\draw[color=black] (70,20,30) -- (100,20,0);
\draw[color=black] (80,30,20) -- (100,30,0);
\draw[color=black] (90,40,10) -- (100,40,0);

\draw[color=black] (100,0,0) -- (100,50,0);
\draw[color=black] (90,0,10) -- (90,40,10);
\draw[color=black] (80,0,20) -- (80,30,20);
\draw[color=black] (70,0,30) -- (70,20,30);
\draw[color=black] (60,0,40) -- (60,10,40);

\draw[color=black] (50,50,0) -- (100,50,0);
\draw[color=black] (40,40,10) -- (90,40,10);
\draw[color=black] (30,30,20) -- (80,30,20);
\draw[color=black] (20,20,30) -- (70,20,30);
\draw[color=black] (10,10,40) -- (60,10,40);

\draw[color=black] (50,50,0) -- (0,0,50);

\draw[color=black] (0,0,50) -- (0,0,50);
\draw[color=black] (10,0,50) -- (10,10,40);
\draw[color=black] (20,0,50) -- (20,20,30);
\draw[color=black] (30,0,50) -- (30,30,20);
\draw[color=black] (40,0,50) -- (40,40,10);
\draw[color=black] (50,0,50) -- (50,50,0);
\draw[color=black] (60,10,40) -- (60,50,0);
\draw[color=black] (70,20,30) -- (70,50,0);
\draw[color=black] (80,30,20) -- (80,50,0);
\draw[color=black] (90,40,10) -- (90,50,0);

\filldraw[color=blue!80] (0,0,50) -- (50,0,50) -- (100,0,0) -- (50,0,0) -- cycle;

\filldraw[color=blue!40] (50,0,50) -- (100,50,0) -- (100,0,0) -- cycle;

\filldraw[color=blue!20] (0,0,50) -- (50,50,0) -- (100,50,0) -- (50,0,50) -- cycle;

\end{axis}

\end{tikzpicture}
    \caption{Nontrivial Region.}
    \label{fig: LI can be nontrivial b}
\end{subfigure}
\caption{Nontrivial non-independent mechanisms can have arbitrarily low influence.}
\label{fig: LI can be nontrivial}
\end{figure*}
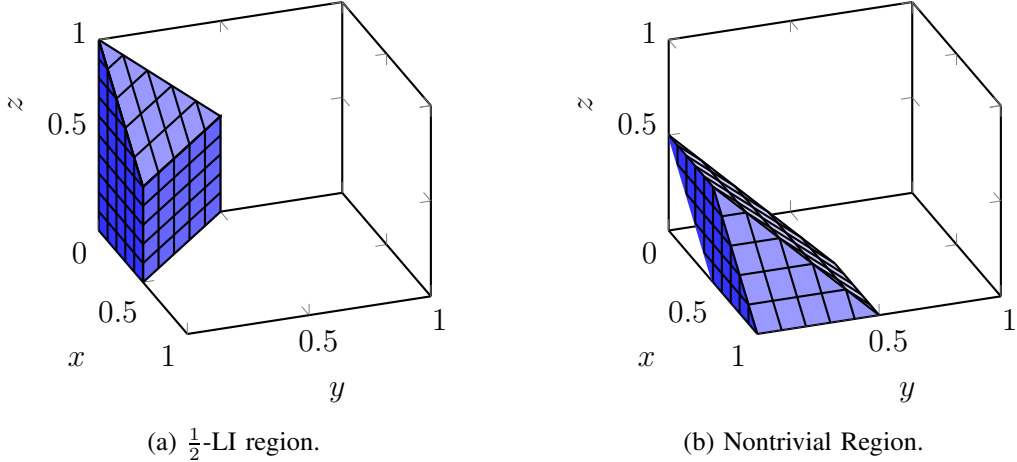

\subsection{Bringing the message home}

The toy example illustrates the interesting relationships and tensions between influence, differential privacy, utility and the independence property for randomized mechanisms. To understand the landscape of randomized mechanisms, we consider general mechanisms, beyond independent mechanisms, and characterize precisely how each notion relates to another.
In particular, we show the following.
\begin{enumerate}
    \item Independent nontrivial mechanisms cannot be $\iota$-LI, for $\iota < \frac{1}{2}$. In particular, additive noise mechanisms often used for differential privacy cannot be simultaneously low influence and useful.
    \item Low influence mechanisms are differentially private. More specifically, $\iota$-LI finite output mechanisms are $(0,\delta)$-DP, for $\delta = \iota(|\mathcal{V}| - 1)$.
    \item There exist non-independent mechanisms which are simultaneously  $(\epsilon,\delta)$-DP, $\iota$-LI, and nontrivial.
\end{enumerate}

\begin{figure}[!t]
    \centering
    \begin{tikzpicture}[every edge/.append style = {draw=blue, line cap=round,line width=2pt}] 

\definecolor{ao(english)}{rgb}{0.0, 0.5, 0.0}
\definecolor{ceruleanblue}{rgb}{0.16, 0.32, 0.75}
\definecolor{cadmiumred}{rgb}{0.89, 0.0, 0.13}

\path (0,0)  edge[ceruleanblue]  (4,0)
(4,0)  edge[cadmiumred]  (2,3.4641) 
(2,3.4641)  edge[ao(english)]  (0,0); 

\node at (0,-0.3) {Nontrivial};
\node at (4,-0.3) {Low Influence};
\node at (2,3.7641) {Independent};

\node at (2,1.5) {Differential};
\node at (2,1) {Privacy};

\end{tikzpicture}

% \begin{tikzpicture}

% \definecolor{ao(english)}{rgb}{0.0, 0.5, 0.0}
% \definecolor{ceruleanblue}{rgb}{0.16, 0.32, 0.75}

% %uncomment this to view grid:
% % \draw[step=0.5cm,black,thin,dashed] (0,0) grid (10,10);

% \draw[line width=2pt,color=black] (0,0) -- (4,0) -- (2,3.4641) -- cycle;
% \draw[line width=2pt,color=red] (0,0) -- (4,0);
% \draw[line width=2pt,color=ceruleanblue] (4,0) -- (2,3.4641);
% \draw[line width=2pt,color=ao(english)] (2,3.4641) -- (0,0);

% \node at (0,-0.3) {Non-trivial};
% \node at (4,-0.3) {Low Influence};
% \node at (2,3.7641) {Independent};

% \node at (2,1.5) {Differential};
% \node at (2,1) {Privacy};

% \end{tikzpicture}
    \caption{A differentially private mechanism can satisfy any two of the properties, but not all three.  }
    \label{fig: triangle}
\end{figure}
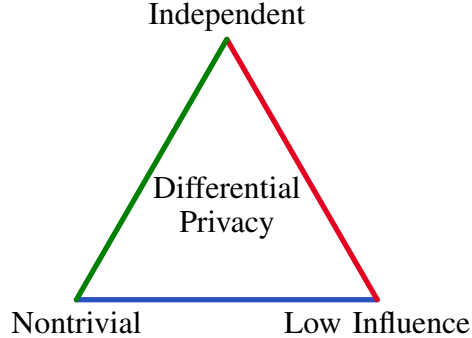

Together, these results form a more complete taxonomy of randomized mechanisms and their properties. An important consequence is that, while independent mechanisms may be sufficient for differential privacy, they fall short when one wants to design mechanisms that satisfy additional properties. 
The low influence condition is one instance that exemplifies the opportunities that arise from studying general non-independent randomized mechanisms. 
In other words, when restricting randomized mechanisms to satisfy additional conditions, e.g. low-influence, fairness, or composition rules, it might be necessary to move away from independent mechanisms to achieve the best performance.

Additionally, it turns out that it is possible to characterize the space of general joint mechanisms, at least for mechanisms with finite inputs and outputs.
While this space has a larger dimension due to the increased number of degrees of freedom, we show that this increased dimensionality may be a blessing, rather than a curse.
For example, we shall show that the $\iota$-LI property can be expressed as a linear constraint over the space of general mechanisms -- a fact that is no longer true if one restricts oneself to independent mechanisms.

The relationship between independence, differential privacy, triviality, and low influence is represented graphically in Figure~\ref{fig: triangle}. 
The green edge signifies independent and useful mechanisms that are prevalent in differential privacy literature (but cannot be arbitrarily low influence). 
The red edge signifies the subset of independent differentially private mechanisms that can be arbitrarily low influence, but are trivial (hence, not useful). 
The results of this paper advocate for further studying non-independent differentially private mechanisms (existing on the blue edge) that can be jointly low influence and useful. For example, non-independent schemes could be useful for differentially private voting schemes. As we showed in the first numerical example, if a private voting scheme is not low influence, then either a single voter changing their vote can change the outcome significantly, or the noise in the mechanism is playing a significant role in the outcome, both undesirable properties.

\subsection{Notation and Problem Setup}

We use $[k]$ to denote the set $\{1,2, \cdots, k\}$.
We denote by $\mathcal{D}$ the space of databases and often enumerate them as $\mathcal{D} = \{d_1 , \ldots , d_{|\mathcal{D}|} \}$. We consider a symmetric neighborhood relationship in $\mathcal{D}$ where $d_i,d_j \in \mathcal{D}$ are said to be neighbors if $d_i \sim d_j$. When clear from context, we use the shorthand notation $d,d'$ to refer to two generic neighboring datasets. In many applications, $\mathcal D =  \mathbb{F}_2^n$ and two databases $d,d' \in \mathbb{F}_2^n$ are neighbors if they differ in one element. 

To avoid degenerate cases, in this paper, we only consider spaces of databases with at least two datasets and which are connected, i.e. for any two $d$ and $d'$ in $\mathcal{D}$, there exists a finite sequence $d = d_1,\ldots,d_k = d'$, such that $d_i \sim d_{i+1}$ for $i \in [1:k-1]$.

We also consider a finite output space $\mathcal{V}$ which corresponds to the space over which the output of the queries lie. A few remarks are in order about finite output mechanisms, which are relatively less studied in the differential privacy literature compared to continuous-output or countable discrete-output mechanisms. Adding continuous noise to an inherently discrete-valued query (such as a counting query) and post processing may make the output less reliable. With the adoption of differential privacy framework for the US Census, there is emerging motivation to either understand existing discrete privacy-preserving mechanisms through the lens of differential privacy \cite{abs:lens} or propose new discrete and finite-valued mechanisms. In this context, the Geometric distribution has been studied in \cite{gm:paper, dp:finite:computers:arxiv, Geometric_Minimax} as the discrete counterpart to the Laplace distribution. The Differential privacy properties of the discrete Gaussian distribution was recently studied in \cite{count:dp:gaussian:arxiv}. However, the support of the noise is the entirety of $\mathbb Z$, which may lead to undesirable  outcomes. For example, according to a recent article in the New York Times \cite{nytimes}, the population of a county in the 2010 US Census was over-reported by a factor of almost 8. Post-processing (truncating or folding the probability mass function) can be used to limit the range of a discrete mechanism \cite{count:dp;arxiv}. However, such an approach may not be tailored to the problem at hand. A recent work \cite{finite:counting} has reported an optimized discrete distribution for counting queries with a finite range as low as $|\mathcal V| = 8$ that can outperform the discrete Gaussian distribution in terms of privacy-utility tradeoff. 

In Table \ref{tab: notation} we list the main symbols used throughout this paper together with their description.

\begin{table}[t]
\centering
\begin{tabular}{p{2cm} p{7.5cm} }  
Symbol & Definition \\
\toprule
$\mathcal{M}$  & A random mechanism.  \\[3pt]
$\mathcal{D}$  & Set of datasets. \\[3pt]
$\mathcal{V}$  & Output set of the mechanism. \\[3pt]
$d \in \mathcal{D}$  & A dataset. \\[3pt]
$v \in \mathcal{V}$  & An output value. \\[3pt]
$(\epsilon,\delta)$  & Differential privacy parameters. \\[3pt]
$\iota$  & Low influence parameter. \\[3pt]
$x,y,z$  & Coordinates of the Euclidean space $\mathbb{R}^{3}$. \\[3pt]
$\mathfrak{I}$ & Set of independent mechanisms.\\[3pt]
$\mathfrak{M}$ & Set of general mechanisms.\\

\bottomrule %\vspace{-5pt}
\end{tabular}
\caption{List of the main symbols used with their descriptions.}
\label{tab: notation}
\end{table}

\section{A Classification of Randomized Mechanisms}
\subsection{Independent versus Joint Mechanisms}
Randomized mechanisms are a key notion in differential privacy. In the literature, these have usually been defined as we do in Definition \ref{def: independent mechanism}. We, however, refer to them as independent mechanisms. The reason for this is that, in Definition \ref{def: joint mechanism}, we present a more general notion of mechanism, which we call joint mechanisms. These more general mechanisms are needed for the results in Sections \ref{sec: LI implies DP} and \ref{sec: nontrivial lowinfluence dp schemes}. 

% In the rest of this paper, we use $M$ to refer to an independent mechanism and $P$ to explicitly refer to a general (including both independent and joint) mechanism. Additionally, we use $\mathcal M: \mathcal{D} \to \mathcal{V}$ for a generic randomized mechanism, whose independence or jointness is not specified. We also use $\mathcal M(d)$ to refer to the resulting random variable when this generic randomized mechanism operates on dataset $d$.

We begin by defining independent mechanisms. To do this, we need the notion of a probability simplex. Given a finite set $\mathcal{V}$, the probability simplex over $\mathcal{V}$ is the set 
\[  \Delta (\mathcal{V}) = \left \{ \lambda = (\lambda_1, \cdots, \lambda_{|\mathcal{V}|})\in \mathbb{R}^{|\mathcal{V}|} : \lambda_i \geq 0 \quad \text{for all $i$, and} \quad \sum_{i=1}^{|\mathcal{V}|} \lambda_i = 1 \right \} .\]

\begin{definition} \label{def: independent mechanism}
An independent mechanism $\mathcal M$ with domain $\mathcal{D}$ and range $\mathcal{V}$ is associated with a mapping $g: \mathcal{D} \rightarrow \Delta (\mathcal{V})$. On input $d \in \mathcal{D}$, the randomized mechanism outputs $v \in \mathcal{V}$ with probability $\Pr[\mathcal M(d) = v]$. We denote the set of independent mechanisms by $\mathfrak{I}$.
\end{definition}

An independent mechanism $\mathcal{M}: \mathcal{D} \rightarrow \mathcal{V}$ is solely determined by the marginal probabilities $M_{ij} \doteq \Pr[\mathcal M(d_i)=v_j]$, which can be arranged in a $|\mathcal D|\times |\mathcal V|$ stochastic matrix. In other words, \[ \mathfrak{I} \cong \left \{ M \in \mathbb{R}^{ |\mathcal{D}| \times |\mathcal{V}|} : M_{ij} \geq 0, \quad  (i,j) \in [|\mathcal{D}|] \times [|\mathcal{V}|],\quad \sum_{j=1}^{|\mathcal{V}|} M_{ij} = 1,  \quad i \in [|\mathcal{D}|] \right \} .\] 
Because of this relation, we often identify $\mathfrak{I}$ with its isomorphic subset of $\mathbb{R}^{ |\mathcal{D}| \times |\mathcal{V}|}$.

Note that independent mechanisms contain both (less common) \emph{query-output-dependent} mechanisms, such as \cite{nissim:smith:query:dependent:2007} (see also \cite{Viswanath_DP_Staircase} for more discussion) and \emph{query-output-independent} mechanisms that are much more commonplace in the differential privacy literature. For discrete queries taking values over a finite field $\mathbb F_q$, we say an (additive) independent mechanism operating on the true query output $f(d)$ is  query-output-independent if for all $u,v \in \mathbb F_q$ and $d\in \mathcal D$, we have $\Pr[\mathcal M(d) = v|f(d) = u] = \Pr[N = v-u]$, where $N$  is an independent noise variable whose distribution does not depend on $u$. We say an (additive) independent mechanism operating on the true query output $f(d)$ is  query-output-dependent if for all $u,v \in \mathbb F_q$ and $d\in \mathcal D$, we have $\Pr[\mathcal M(d) = v|f(d) = u] = \Pr[N_u = v-u]$, where $N_u$ is an independent noise variable whose distribution depends on $u$. We refer to both such mechanisms as independent because, as for all independent mechanisms, the random variables $\{\mathcal M(d) : d \in \mathcal{D} \}$ are mutually independent. 

This is not the case for the joint mechanisms we now define.

%\begin{definition}
%We denote the space of all functions from $\mathcal{D}$ to $\mathcal{V}$ by $\mathcal{V}^{\mathcal{D}} = \{ f: \mathcal{D} \rightarrow \mathcal{V} \}$.
%\end{definition}

%We now define joint mechanisms.

%\begin{definition} \label{def: joint mechanism} \raf{I do not think this characterization is correct. The function $P: \mathcal{D} \rightarrow \Delta (\mathcal{V}^{|\mathcal{D}|})$. Let me know if I misunderstand it. I left definition 5 as a suggestion.}
%A joint mechanism $P$ with domain $\mathcal{D}$ and range $\mathcal{V}^{|\mathcal{D}|}$ is associated with a mapping $P: \mathcal{D} \rightarrow \Delta (\mathcal{V}^{|\mathcal{D}|})$. We denote the set of joint mechanisms by $\mathfrak{M}$. A joint mechanism $P \in \mathfrak{M}$ is determined by the joint probabilities $\Pr[\mathcal{M}(d_1)=v_1 , \ldots, \mathcal{M}(d_{|\mathcal{D}|})=v_{|\mathcal{D}|}]$ for every $(v_1, \ldots, v_{|\mathcal{D}|}) \in \mathcal{V}^{|\mathcal{D}|}$. That is, a general joint mechanism is equivalent to probability distributions on $\mathcal{V}^{\mathcal{D}}$. Thus, 
%\[ \mathfrak{M} \cong \left \{ P_{v_1,\ldots,v_{|\mathcal{D}|}} \in \mathbb{R}^{{|\mathcal{V}|}^{|\mathcal{D}|}} : P_{v_1,\ldots,v_{|\mathcal{D}|}} \geq 0 \quad \text{for every $v_i \in \mathcal{V}$ and} \quad \sum_{(v_1,\ldots,v_{|\mathcal{D}|}) \in \mathcal{V}^{|\mathcal{D}|}} P_{v_1,\ldots,v_{|\mathcal{D}|}} = 1 \right \} .\]
%\end{definition}

\begin{definition} \label{def: joint mechanism}
A joint randomized mechanism $\mathcal{M}:\mathcal{D} \rightarrow \mathcal{V}$ is defined by the joint probabilities $\Pr[\mathcal{M}(d_1)=v_1 , \ldots, \mathcal{M}(d_{|\mathcal{D}|})=v_{|\mathcal{D}|}]$ for every $(v_1, \ldots, v_{|\mathcal{D}|}) \in \mathcal{V}^{|\mathcal{D}|}$. We denote the set of joint mechanisms by $\mathfrak{M}$.
\end{definition}

A joint mechanism $\mathcal{M}: \mathcal{D} \rightarrow \mathcal{V}$ is solely determined by the joint probabilities of the form $P_{v_1,\ldots,v_{|\mathcal{D}|}} \doteq \Pr[\mathcal{M}(d_1)=v_1 , \ldots, \mathcal{M}(d_{|\mathcal{D}|})=v_{|\mathcal{D}|}]$. Thus,
\begin{multline} \label{eq: space of joint}
\mathfrak{M} \cong \Bigg \{ P_{v_1,\ldots,v_{|\mathcal{D}|}} \in \mathbb{R}^{{|\mathcal{V}|}^{|\mathcal{D}|}} : P_{v_1,\ldots,v_{|\mathcal{D}|}} \geq 0, \quad \text{for every $(v_1,\ldots,v_{|\mathcal{D}|}) \in \mathcal{V}^{|\mathcal{D}|}$ and} \\\quad \sum_{(v_1,\ldots,v_{|\mathcal{D}|}) \in \mathcal{V}^{|\mathcal{D}|}} P_{v_1,\ldots,v_{|\mathcal{D}|}} = 1 \Bigg \} .
\end{multline}
Because of this relation, we often identify $\mathfrak{M}$ with its isomorphic subset of $\mathbb{R}^{{|\mathcal{V}|}^{|\mathcal{D}|}}$.

Clearly the set of joint mechanisms contains the class of independent mechanisms as a special case.  That is, for an independent mechanism and for every $(v_1, \ldots, v_{|\mathcal{D}|}) \in \mathcal{V}^{|\mathcal{D}|}$, we have 
\[ \Pr[\mathcal{M}(d_1)=v_1 , \ldots, \mathcal{M}(d_{|\mathcal{D}|})=v_{|\mathcal{D}|}] = \prod_{i=1}^{|\mathcal D|} \Pr[\mathcal M(d_i)=v_i]. \]

We have already identified the space of independent mechanisms $\mathfrak{I}$ with a subset of $\mathbb{R}^{|\mathcal{D}| \times |\mathcal{V}|}$, and the space of general mechanisms $\mathfrak{M}$ with a subset of $\mathbb{R}^{{|\mathcal{V}|}^{|\mathcal{D}|}}$. Since an independent mechanism is a special case of joint mechanism, it is clear that there must exist a subset of $\mathfrak{M}$ isomorphic to  $\mathfrak{I}$. We characterize this subset in the following proposition.

\begin{proposition} \label{pro: surface of independent mechanisms}
It holds that,
\begin{multline*}
    \mathfrak{I} \cong \Bigg \{ P_{v_1,\ldots,v_{|\mathcal{D}|}} \in \mathfrak{M} : P_{v_1,\ldots,v_{|\mathcal{D}|}} = \prod_{i=1}^{|\mathcal{D}|} \sum_{({w_1,\ldots,w_{|\mathcal{D}|}}) \in \mathcal{V}^{|\mathcal{D}|} : w_i = v_i }  P_{w_1,\ldots,w_{|\mathcal{D}|}} , \\ \text{for every $(v_1,\ldots,v_{|\mathcal{D}|}) \in \mathcal{V}^{|\mathcal{D}|}$} \Bigg \} .
\end{multline*}
In other words, $\mathfrak{I}$ is a polynomial of degree $|\mathcal{D}|$ in $\mathbb{R}^{{|\mathcal{V}|}^{|\mathcal{D}|}}$.
\end{proposition}

\begin{proof}
See the Appendix.
\end{proof}

Not to confuse the previous identifications, we refer to the space in the above proposition as $\mathfrak{I}_{\mathfrak{M}}$. The statement of the proposition is then equivalent to $\mathfrak{I} \cong \mathfrak{I}_{\mathfrak{M}} \subset \mathfrak{M}$. Throughout the paper, when we refer to a mechanism without specification, we mean a joint mechanism as these include independent ones.

We finish this section by showing a different characterization of randomized mechanisms and giving an explicit example of a non-independent mechanism.

\begin{proposition} \label{pro: mech prob}
A randomized mechanism is equivalent to a probability distribution on the space of all functions $h: \mathcal{D} \rightarrow \mathcal{V}$.
\end{proposition}

\begin{proof}
See the Appendix.
\end{proof}

This proposition gives an intuitive way of constructing randomized mechanisms by randomly choosing a function $h: \mathcal{D} \rightarrow \mathcal{V}$, as we do in the following example.

\begin{example} \label{ex: nonindep mech}

Consider the case where $\mathcal{D}$ consists of two neighboring datasets ${d_1 \sim d_2}$, and the output space is binary, i.e., $\mathcal{V} = \{1,2 \}$. Then, the set of all functions $h: \mathcal{D} \rightarrow \mathcal{V}$ is given by $\mathcal{V}^\mathcal{D} = \{ f_{11}, f_{12}, f_{21}, f_{22}\}$, where $f_{ij}: \mathcal{D} \rightarrow \mathcal{V}$ is such that 
\begin{align*}
   f_{ij}(d_k) = \left\{\begin{matrix}
i & \text{if $k=1$},\\ 
j & \text{if $k=2$}.
\end{matrix}\right.
\end{align*}

By proposition \ref{pro: mech prob}, a randomized mechanism is equivalent to a probability distribution over $\mathcal{V}^\mathcal{D}$. We denote this distribution by a $2\times 2$ matrix $P$ such that $P_{ij}$ is the probability of choosing the function $f_{ij}$. We define a random variable $F$ to denote this random function, thus $P_{ij} = \Pr[F=f_{ij}]$. The random mechanism is then obtained by randomly choosing a function $f_{ij} \in \mathcal{V}^\mathcal{D}$ according to $P$ and evaluating it on whatever dataset one has. Thus, 
\begin{align*}
    \Pr[\mathcal{M}(d_1)=i,\mathcal{M}(d_2)=j] = \Pr[F(d_1)=i,F(d_2)=j] = \Pr[F=f_{ij}] = P_{ij}.
\end{align*}
Thus, the mechanism $\mathcal{M}$ is independent if and only if $P_{ij} = (P_{i1}+P_{i2})(P_{1j}+P_{2j})$.

Revisiting the first numerical example in the introduction, let us denote by $d_1, d_2$ two datasets where the last individual voted YES (outcome 1) and NO (outcome 2), respectively. We use the following probability matrix for choosing $f_{ij}$:
\begin{align}
P = \left[\begin{matrix}
     0.45 &0.1\\0&0.45
     \end{matrix}\right]
\end{align}
It is clear that this mechanism is not independent and the probability of giving different answers due to the vote of the last individual is capped at $P(\mathcal{M}(d_1) \neq \mathcal{M}(d_2)) = P(\mathcal{M}(d_1) = 1, \mathcal{M}(d_2) = 2) = P(F = f_{12}) = P_{12} = 0.1$, while the probability of being truthful in both YES and NO majority cases is $P(\mathcal{M}(d_1) = 1) = P(F = f_{11}) + P(F = f_{12}) = P(\mathcal{M}(d_2) = 2) = P(F = f_{12})+ P(F = f_{22}) = 0.55$, resulting in the same $\epsilon \approx 0.2$ as before. Note that we do not have to have a symmetric randomized response matrix to ensure a mechanism which does not favor one majority outcome over another.
\end{example}

\subsection{Differential Privacy, Low Influence, and Nontrivial Mechanisms}

In this section we define different families of mechanisms. 

Recall that differential privacy was defined in Definition \ref{def: dp}. We also use the following notion.

% \begin{definition} \label{def: dp}
% A mechanism $\mathcal{M}: \mathcal{D} \to \mathcal{V}$ is  $(\epsilon,\delta)$-DP if, for any two neighboring datasets $d \sim d'$, we have $\Pr [\mathcal{M}(d) \in \mathcal{S}] \leq e^\epsilon \Pr [\mathcal{M}(d') \in \mathcal{S}] + \delta$, for every $S \subseteq \mathcal{V}$.
% \end{definition}

\begin{definition} \label{def: vdp}
A mechanism $\mathcal{M}: \mathcal{D} \to \mathcal{V}$ is $(\epsilon,\delta)$-value differentially private, or $(\epsilon,\delta)$-VDP in short, if for any two neighboring $d \sim d'$, we have $\Pr [\mathcal{M}(d) = v] \leq e^\epsilon Pr [\mathcal{M}(d') = v] + \delta$, for every $v \in \mathcal{V}$.
\end{definition}

A mechanism which is value differentially private is also differentially private.

\begin{proposition} \label{pro: vdp equiv dp}
Let the mechanism $\mathcal{M}$ be $(\epsilon, \delta)$-VDP. Then, $\mathcal{M}$ is $(\epsilon,(|\mathcal{V}|-1) \delta)$-DP.
\end{proposition}

\begin{proof}
Since $\mathcal{M}$ is $(\epsilon, \delta)$-VDP, it follows that $\Pr[\mathcal{M}(d)=v] \leq e^\epsilon \Pr[\mathcal{M}(d')=v] + \delta$, for every neighboring datasets $d \sim d'$. But then, for every neighboring datasets $d \sim d'$,
\begin{align*}
    \Pr[\mathcal{M}(d) \in \mathcal{S}] &= \sum_{v \in \mathcal{S}} \Pr[\mathcal{M}(d) = v] \leq \sum_{v \in \mathcal{S}} \left( e^\epsilon \Pr[\mathcal{M}(d') = v] + \delta \right) \\
    &= e^\epsilon \Pr[\mathcal{M}(d') \in \mathcal{S}] + |\mathcal{S}| \delta.
\end{align*}
Since the differential privacy condition trivially holds for $\mathcal{S} = \mathcal{V}$, we can upper bound $|\mathcal{S}|$ by $|\mathcal{V}|-1$, proving the result.
\end{proof}

In Definition \ref{def: low influence mechanism}, we present the notion of low influence mechanisms. This notion is an adaptation of the following definition, widely used in social choice theory and in the analysis of Boolean functions.

\begin{definition} \label{def: original influence}
The influence of coordinate $i$ on a function $f:\{ -1 ,1 \}^n \rightarrow \{-1 , 1 \}$ is defined as the probability $\Inf_i [f] = \Pr [f(x) \neq f(x^{\oplus i})]$, where $x$ is a uniformly distributed vector from $\{ -1 ,1 \}^n$, and $x^{\oplus i}$ is the vector such that $x^{\oplus i}_i = - x_i$  and $x^{\oplus i}_k = x_k$ for $k \in [n]-\{i\}$.
\end{definition}

If one thinks of $f$ as a voting rule, i.e., a rule for determining how to interpret $n$ votes for two candidates $\{ -1 ,1 \}$ as an election of one of them, then the influence $\Inf_i [f]$ measures the probability that the $i$-th voter will affect the outcome of the election. In this context, we say a function $f$ is $\iota$-low influence if $\Inf_i [f] \leq \iota$ for every $i \in [n]$.

We note that the probability in the above definition is taken over the argument of the function, and that the function itself is deterministic. The opposite is true in the differential privacy setup. The definition of low influence in Definition \ref{def: low influence mechanism} in the introduction accounts for this in measuring the influence through $\Pr[\mathcal{M}(d) \neq \mathcal{M}(d')] \leq \iota$, where the randomness is in the randomized mapping $\mathcal{M}$, not the inputs $d$ and $d'$.

% To account for this, we propose the following adaptation.

% \begin{definition} \label{def: low influence mechanism}
% A mechanism $\mathcal{M}: \mathcal{D} \to \mathcal{V}$ is $\iota$-low influence, or $\iota$-LI in short, if for any two neighboring datasets $d \sim d'$, we have $\Pr [\mathcal{M}(d) \neq \mathcal{M}(d')] \leq \iota$.
% \end{definition}

Finally, we define nontrivial mechanisms as a randomized generalization of non-constant mechanisms. More precisely, we say a mechanism is nontrivial if its most likely output is not constant across datasets.

\begin{definition}\label{def:nontrivial} A mechanism $\mathcal{M}: \mathcal{D} \to \mathcal{V}$ is nontrivial if there exists neighboring datasets $d \sim d'$ such that $\argmax_v \Pr[\mathcal{M}(d) = v] \neq  \argmax_v \Pr[\mathcal{M}(d') = v]$.
\end{definition}

Note this definition of nontriviality is very mild. In practice, the utility of a mechanism is application dependent, and the mechanism being nontrivial is not sufficient, but rather necessary. Recall the voting example, where to nontrivially approximate the majority function, the most likely outcome must be different across two marginal datasets. We saw this can result in high $\iota$ influence in independent mechanisms, but can be controlled via joint mechanisms. Across non-marginal datasets, the most likely outcome may be identical without being detrimental to approximating the true majority function, DP or LI constraints.

%I think this example will be too big.

% \rev{In the following example we show the different notions defined in this section.}

% \begin{example}
% Consider the non-independent mechanism of example \ref{ex: nonindep mech} given by $P = \left( \begin{smallmatrix} 0.33 & 0.34\\  0 & 0.33 \end{smallmatrix} \right )$. This mechanism is $(\log(2),0.01)$-VDP. Indeed, the VDP bounds are tight since
% \[ \Pr[\mathcal{M}(d_1) = 1] = P_{11}+P_{12} = 0.67 = 2 \times 0.33+0.01 = e^{\log(2)} \times \Pr[\mathcal{M}(d_2) = 1] + 0.01 \]

% \end{example}

\section{Nontrivial Independent Mechanisms are Not Low Influence} \label{sec: nontrivial independent not LI}

In this section, we show that if an independent mechanism is nontrivial, then it cannot be low influence. We begin by showing that although differential privacy is equivalent to a set of linear constraints on $\mathfrak{I}$, low influence is equivalent to a set of quadratic constraints.

\begin{proposition}
Let $\mathcal{M}: \mathcal{D} \to \mathcal{V}$ be an independent mechanism with matrix representation $M_{ij} = \Pr[\mathcal{M}(d_i)=v_j]$. Then, $\mathcal{M}$ is $(\epsilon,\delta)$-DP if and only if, for every $d_i \sim d_k$,
\begin{align*}
    \sum_{j:v_j \in \mathcal{S}} M_{ij} \leq e^\epsilon \sum_{j:v_j \in \mathcal{S}}M_{kj} + \delta \quad \text{for every $\mathcal{S} \subseteq \mathcal{V}$.}
\end{align*}
Similarly, $M$ is $(\epsilon,\delta)$-VDP if and only if, for every $d_i \sim d_k$, it holds that $M_{ij} \leq e^\epsilon M_{kj}+\delta$, for every $j \in [|\mathcal{V}|]$. In other words, both $(\epsilon,\delta)$-DP and $(\epsilon,\delta)$-VDP are linear constraints in the space $\mathfrak{I}$.
\end{proposition}

\begin{proof}
This follows directly from substituting $M_{ij} = \Pr[\mathcal{M}(d_i)=v_j]$ into the definitions of $(\epsilon,\delta)$-DP and $(\epsilon,\delta)$-VDP.
\end{proof}

\begin{proposition} \label{pro: independent low influence}
Let $\mathcal{M}: \mathcal{D} \to \mathcal{V}$ be an independent mechanism with matrix representation $M_{ij} = \Pr[\mathcal{M}(d_i)=v_j]$. Then $\mathcal{M}$ is $\iota$-LI if and only if, for every $d_i \sim d_k$,
\begin{align*}
    \sum_{j=1}^{|\mathcal{V}|} M_{ij} M_{kj} \geq 1- \iota .
\end{align*}
In other words, $\iota$-LI conditions are quadratic constraints over the space $\mathfrak{I}$.
\end{proposition}

\begin{proof}
By the definition of $\iota$-LI, for every $d_i \sim d_k$,
\begin{align*}
\begin{split}
    \iota &\geq \Pr[\mathcal M(d_i) \neq \mathcal M(d_k) ] = 1 - \Pr[\mathcal  M(d_i) = \mathcal M(d_k) ] \\
    &= 1 - \sum_{j=1}^{|\mathcal{V}|} \Pr[\mathcal M(d_i)=j , \mathcal M(d_k)=j] \\
    &= 1 - \sum_{j=1}^{|\mathcal{V}|} \Pr[\mathcal M(d_i)=j] \Pr[ \mathcal M(d_k)=j] \\
    &= 1 - \sum_{j=1}^{|\mathcal{V}|} M_{ij} M_{kj}.
\end{split}
\end{align*}
\end{proof}

Before we prove the main result of this section, let us revisit the toy example from the Introduction. We show that for our toy example, a nontrivial mechanism is not $\iota$-LI for $\iota < \frac{1}{2}$.

\begin{example} \label{ex: nontrivial indep not li}
Consider the case where $\mathcal{D}$ consists of two neighboring datasets ${d_1 \sim d_2}$, and the output space is binary, i.e. $\mathcal{V} = \{1,2 \}$. To be consistent with the simpler notation in the toy example in the Introduction, let $x \doteq M_{11} = \Pr[\mathcal M(d_1)=1]$ and $y \doteq M_{21} = \Pr[\mathcal M(d_2)=1]$. Suppose $M$ is nontrivial. Then, either (i) $\frac{1}{2}\leq x \leq 1$ and $0 \leq y \leq \frac{1}{2}$ or (ii) $0 \leq x \leq \frac{1}{2}$ and $ \frac{1}{2} \leq y \leq 1$. Recall that we additionally assume $x$ and $y$ cannot simultaneously be equal to $\frac{1}{2}$. Without loss of generality, we consider case (i) and denote 
\begin{align} \label{eq: example R*}
    R^* = \left \{(x,y)\in \mathbb{R}^2 : \frac{1}{2} \leq x \leq 1 \quad \text{and} \quad 0 \leq y \leq \frac{1}{2} \right \} .
\end{align}
Let $I:\mathbb{R}^2 \rightarrow \mathbb{R}$ be given by $I(x,y)=x+y-2xy$. Then, a mechanism is $\iota$-LI if and only if $I(x,y) \leq \iota$. Consider the following constrained quadratic optimization problem.
\begin{equation} \label{eq: example independent mechanisms}
\begin{aligned}
& {\text{minimize}}
& & I(x,y) = x+y-2xy \\
& \text{subject to}
& & \frac{1}{2} \leq x \leq 1 ,\\
& & & 0 \leq y \leq \frac{1}{2} .
\end{aligned}
\end{equation}
The solution to this problem is the minimum influence that a nontrivial mechanism can have.

We show that the solution to this problem is the set
\[ S = \left \{ \left( \frac{1}{2},t \right) : t \in \left [0,\frac{1}{2} \right [ \hspace{2pt} \right \} \bigcup \left\{ \left( t,\frac{1}{2} \right) : t \in \left ] \frac{1}{2},1 \right ] \right \} ,\] 
with $I(S)=\frac{1}{2}$. To do this, we look at the direction of fastest decrease, $- \nabla I = (2y-1,2x-1)$. Note that if $(x,y) \in R^*-S$ then $ - ( \nabla I)_x < 0$ and $ - ( \nabla I)_y > 0$. This means that, for every point $(x,y) \in R^*$ the direction $- \nabla I (x,y)$ points towards $S$. Thus, the solution space of the optimization problem $\eqref{eq: example independent mechanisms}$ must be contained in $S$. A direct calculation shows that $I(S)=\frac{1}{2}$, concluding that nontrivial mechanisms cannot have influence smaller than $\frac{1}{2}$. We illustrate this in Fig. \ref{fig: independent low influence is trivial b}.
\end{example}

The lower bound in the example above holds true in general.

\begin{theorem} \label{teo: nontrivial indep not li}
Let $\mathcal{M}: \mathcal{D} \to \mathcal{V}$ be an independent, $\iota$-LI, nontrivial mechanism. Then, $\iota \geq \frac{1}{2}$. Furthermore, this bound is tight.
\end{theorem}

\begin{proof}
We present the full proof in the Appendix. The main idea behind the proof is to characterize the set of nontrivial schemes $R$ and show that the minimum influence in that set is equivalent to the minimum influence in the set $R^*$ of Example $\ref{ex: nontrivial indep not li}$. Thus, $\iota \geq \frac{1}{2}$.
\end{proof}

We note that we can obtain a looser lower-bound on the influence of independent nontrivial mechanisms using a simpler argument. For this, we use the mechanism support size $|\mathcal V|$ and the nontrivial mechanism assumption to find a simple lower bound on $\min_{d \in \mathcal{D}} \max_{v \in [\mathcal{V}]} \Pr[\mathcal M(d)=v]$. However, this bound becomes loose as $|\mathcal V|$ increases.

\begin{proposition}
Let $\mathcal{M}: \mathcal{D} \to \mathcal{V}$ be independent, nontrivial, and $\iota$-LI. Then, $\iota \geq \frac{1}{|\mathcal{V}|^2}$.
\end{proposition}

\begin{proof}
Since $\mathcal{M}$ is nontrivial, there exists neighboring datasets $d_1 \sim d_2$ and two different values $v_1,v_2 \in \mathcal{V}$ such that $\Pr[\mathcal M(d_1) = v_1] \geq 1/|\mathcal{V}|$ and $\Pr[\mathcal M(d_2) = v_2] \geq 1/|\mathcal{V}|$. Then,
\begin{align*}
    \Pr[\mathcal M(d_1) \neq \mathcal M(d_2)] &\geq \Pr[\mathcal M(d_1)=v_1 , \mathcal M(d_2) = v_2] \\
    &= \Pr[ \mathcal M(d_1)=v_1 ] \Pr[ \mathcal M(d_2) = v_2] \\
    &\geq \frac{1}{|\mathcal{V}|} \frac{1}{|\mathcal{V}|} = \frac{1}{|\mathcal{V}|^2} .
\end{align*}
\end{proof}

The bound in Theorem \ref{teo: nontrivial indep not li} is strictly larger than the bound above for $|\mathcal{V}|>2$. The following bound, however, can sometimes outperform the bound in Theorem \ref{teo: nontrivial indep not li} and is not restricted only to nontrivial mechanisms.

\begin{proposition}
Let the mechanism $\mathcal{M}: \mathcal{D} \to \mathcal{V}$ be  independent and $\iota$-LI. Then, 
\[ \iota \geq 1 - \min_{d \in \mathcal{D}} \max_{v \in \mathcal{V}} \Pr[\mathcal M(d)=v] .\]
\end{proposition}
First, we define a function which measures the local influence of any two datasets.

\begin{definition} \label{def: influence function}
Let $\mathcal{V}$ be a finite set. We define
$I: \mathbb{R}^{|\mathcal{V}| \times |\mathcal{V}|} \rightarrow \mathbb{R}$, such that $I(u,v) = 1 - \sum_{i=1}^{\mathcal{V}} u_i v_i$.
\end{definition}

\begin{proof}
Let $I$ be the function in Definition \ref{def: influence function}. For the independent mechanism $\mathcal{M}$ with matrix representation $M_{ij} = \Pr[\mathcal{M}(d_i)=v_j]$ and rows $M_i$ and $M_k$, we can write
\begin{align*}
    I(M_i,M_k) &= 1 - \sum_{j=1}^{|\mathcal{V}|} M_{ij} M_{kj} \geq 1 - \min \left( \max_{j \in [|\mathcal{V}|]} M_{ij}  , \max_{j \in [|\mathcal{V}|]}  M_{kj} \right) \\
    &\geq 1 - \min_{l \in [|\mathcal{D}|]}  \max_{j \in [|\mathcal{V}|]}  M_{lj} = 1 - \min_{d \in \mathcal{D}} \max_{v \in [\mathcal{V}]} \Pr[\mathcal M(d)=v].
\end{align*}
But by Lemma 1 in the Appendix, $\mathcal{M}$ is $\iota$-LI if and only if $I(M_i,M_k) \leq \iota$ for every $d_i \sim d_k$. Thus,
\[ \iota \geq 1 - \min_{d \in \mathcal{D}} \max_{v \in \mathcal{V}} \Pr[\mathcal M(d)=v] .\]
\end{proof}

This bound can be better than the one in Theorem \ref{teo: nontrivial indep not li}, if $\min_{d \in \mathcal{D}} \max_{v \in \mathcal{V}} \Pr[\mathcal{M}(d)=v] < \frac{1}{2}$. At best, $\min_{d \in \mathcal{D}} \max_{v \in \mathcal{V}} \Pr[\mathcal{M}(d)=v] = \frac{1}{|\mathcal{V}|}$. At worst, it could be arbitrarily close to the trivial bound of $1$. Thus, although it can sometimes outperform the bound in Theorem \ref{teo: nontrivial indep not li}, it cannot guarantee that the mechanism is bounded away from $0$, as Theorem \ref{teo: nontrivial indep not li} does.

\begin{remark}
In \cite{dwork2015preserving}, the authors studied statistical stability of (independent) differentially private mechanisms in dealing with adaptive data queries. The main question addressed in \cite{dwork2015preserving} is how many adaptive statistical queries can one reliably answer using  \emph{the same} $n$ data samples which are drawn i.i.d from an underlying data distribution $\mathcal{P}$.
 
Recall that for a differentially private mechanism, $\Pr[\mathcal{M}(d)\in \mathcal S]$ does not change fast between neighboring datasets as a function of $d$. This fact is used in \cite{dwork2015preserving} to prove that the empirical average of a randomly chosen hypothesis $\phi$ that is adaptively selected \emph{after} observing $\mathcal{M}(d)$ will be close to the true expectation of $\phi$ with \emph{high probability}. In other words, the chance of observing ``unlucky nontypical" datasets is small, where by unlucky we mean the empirical average of a function dependent on $d$ is far away from its true expectation. This can be thought of as a notion of \emph{typical} invariance or stability to particular realizations of the dataset.
 
We remark that our results in this paper are not inconsistent with those in \cite{dwork2015preserving}. The notion of influence in this paper is different from \cite{dwork2015preserving} and is a worst-case measure: in the sense that the bound $\iota$ on the influence must be universally valid for \emph{all} neighboring datasets. Essentially, the fact that $\iota$ for independent mechanisms is bounded away from 0, as shown in Theorem \ref{teo: nontrivial indep not li} means that in order to preserve utility \emph{there exists} \emph{at least two (and possibly many more)} neighboring datasets which are likely to give very different results to the same query. That is, for independent mechanisms there exists ``unlucky nontypical datasets" which if hit by chance, can lead to unreliable statistical measures.

In Section \ref{sec: nontrivial lowinfluence dp schemes}, we show that this issue can be overcome through designing mechanisms on the joint space of datasets. That is, nontrivial mechanisms can be universally low influence if the mechanisms are non-independent.
\end{remark}

\section{Low Influence Mechanisms are Differentially Private} \label{sec: LI implies DP}

In this section, we show that low influence functions are differentially private. We will not restrict ourselves to only independent mechanism, considering joint mechanisms too. We begin by showing that any linear constraint on $\mathfrak{I}$ is also linear on $\mathfrak{M}$.

\begin{proposition} \label{pro: linear in I is linear in M}
If a constraint is linear on $\mathfrak{I}$, it is linear on $\mathfrak{M}$. In particular, both $(\epsilon,\delta)$-DP and $(\epsilon,\delta)$-VDP are linear on $\mathfrak{M}$.
\end{proposition}

\begin{proof}
Let $P \in \mathfrak{M}$ with $P_{v_1, \ldots, v_{|\mathcal{D}|}} = \Pr[\mathcal{M}(d_1)=v_1, \ldots \mathcal{M}(d_{|\mathcal{D}|})=v_{|\mathcal{D}|}]$ and $M \in \mathfrak{I}$ with $M_{kj}= \Pr[\mathcal{M}(d_k)=v_j]$. Then, 
\[ M_{kj}= \Pr[\mathcal{M}(d_k)=v_j] = \sum_{({w_1,\ldots,w_{|\mathcal{D}|}}) \in \mathcal{V}^{|\mathcal{D}|} : w_k = v_j }  P_{w_1,\ldots,w_{|\mathcal{D}|}} .\]
Thus, each $M_{kj}$ is a linear combination of $P_{v_1, \ldots, v_{|\mathcal{D}|}}$.
\end{proof}

Thus, looking at $\mathfrak{M}$ does not complicate the differential privacy constraints. Indeed, it will also simplify the low influence constraints into linear constraints in the space $\mathfrak{M}$ (which were non-linear in the space $\mathfrak{I}$). Proposition \ref{pro: surface of independent mechanisms} explains why the low influence constraints are linear on $\mathfrak{M}$, but non-linear (quadratic) on $\mathfrak{I}$. This is illustrated in Figure \ref{fig: general to independent} for $|\mathcal{V}| = 2$.

\begin{figure*}[!t]
\begin{subfigure}[t]{0.30\textwidth}
    \centering
    \begin{tikzpicture}

\begin{axis}[width=5cm, height=5cm,stack plots=y,thick,smooth,no markers,xmin=0,xmax=1,ymin=0,ymax=1,zmin=0,zmax=1,xlabel={$x$},
  ylabel={$y$},zlabel={$z$},xtick={0,0.5,1},ytick={0.5,1},ztick={0.5,1},view={60}{30}]

\draw[color=black] (0,0,10) -- (50,0,10);
\draw[color=black] (0,0,20) -- (50,0,20);
\draw[color=black] (0,0,30) -- (50,0,30);
\draw[color=black] (0,0,40) -- (50,0,40);
\draw[color=black] (0,0,50) -- (50,0,50);
\draw[color=black] (0,0,60) -- (40,0,60);
\draw[color=black] (0,0,70) -- (30,0,70);
\draw[color=black] (0,0,80) -- (20,0,80);
\draw[color=black] (0,0,90) -- (10,0,90);

\draw[color=black] (10,0,0) -- (10,0,90);
\draw[color=black] (20,0,0) -- (20,0,80);
\draw[color=black] (30,0,0) -- (30,0,70);
\draw[color=black] (40,0,0) -- (40,0,60);
\draw[color=black] (50,0,0) -- (50,0,50);

\draw[color=black] (0,0,100) -- (50,0,50);
\draw[color=black] (0,10,90) -- (40,10,50);
\draw[color=black] (0,20,80) -- (30,20,50);
\draw[color=black] (0,30,70) -- (20,30,50);
\draw[color=black] (0,40,60) -- (10,40,50);

\draw[color=black] (10,0,90) -- (0,10,90);
\draw[color=black] (20,0,80) -- (0,20,80);
\draw[color=black] (30,0,70) -- (0,30,70);
\draw[color=black] (40,0,60) -- (0,40,60);
\draw[color=black] (50,0,50) -- (0,50,50);

\draw[color=black] (50,0,50) -- (50,0,0);
\draw[color=black] (40,10,50) -- (40,10,0);
\draw[color=black] (30,20,50) -- (30,20,0);
\draw[color=black] (20,30,50) -- (20,30,0);
\draw[color=black] (10,40,50) -- (10,40,0);

\draw[color=black] (50,0,40) -- (0,50,40);
\draw[color=black] (50,0,30) -- (0,50,30);
\draw[color=black] (50,0,20) -- (0,50,20);
\draw[color=black] (50,0,10) -- (0,50,10);

\draw[color=black] (0,0,0) -- (0,0,100) -- (50,0,50) -- (50,0,0) -- cycle;

\draw[color=black] (50,0,0) -- (50,0,50) -- (0,50,50) -- (0,50,0) -- cycle;

\draw[color=black] (0,0,100) -- (50,0,50) -- (0,50,50) -- cycle;

\filldraw[color=blue!80] (0,0,0) -- (0,0,100) -- (50,0,50) -- (50,0,0) -- cycle;

\filldraw[color=blue!60] (50,0,0) -- (50,0,50) -- (0,50,50) -- (0,50,0) -- cycle;

\filldraw[color=blue!40] (0,0,100) -- (50,0,50) -- (0,50,50) -- cycle;

\end{axis}

\end{tikzpicture}
    \caption{Low influence in the space of general mechanisms.}
    \label{fig: general to independent a}
\end{subfigure}
\begin{subfigure}[t]{0.30\textwidth}
    \centering
    % Commented for quicker compilation.
    \input{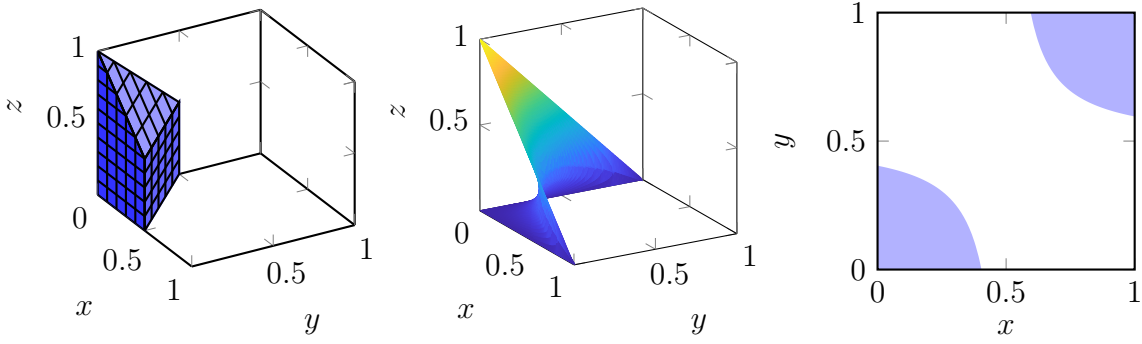}
    \caption{Independence surface in the space of general mechanisms.}
    \label{fig: general to independent b}
\end{subfigure}
\begin{subfigure}[t]{0.30\textwidth}
    \centering
    \begin{tikzpicture}
%uncomment this to view grid:
%\draw[step=0.5cm,black,thin,dashed] (0,0) grid (10,10);

\begin{axis}[width=5cm, height=5cm,
stack plots=y,thick,smooth,no markers,xmin=0,xmax=1,ymin=0,ymax=1,xlabel={$x$},
  ylabel={$y$},axis on top,xtick={0,0.5,1},ytick={0,0.5,1}]
  
\draw[scale=100,fill=blue!30,color=blue!30,domain=0:0.498,variable=\x] plot ({\x},{(5*\x-2)/(10*\x-5)});

\draw[scale=100,fill=blue!30,color=blue!30,domain=0.503:1,variable=\x] plot ({\x},{(5*\x-2)/(10*\x-5)});

\filldraw[fill=blue!30, draw=blue!30] (0,0) -- (0,40) -- (20,0)-- cycle;

\filldraw[fill=blue!30, draw=blue!30] (100,100) -- (100,60) -- (80,100)-- cycle;
\end{axis}
\end{tikzpicture}
    \caption{Low influence in the space of independent mechanisms.}
    \label{fig: general to independent c}
\end{subfigure}
\caption{For $|\mathcal{V}| = 2$, Figure (a) shows the region for low influence in the space of general (joint) mechanisms. Figure (b) shows the surface of independent mechanisms. Figure (c) can be interpreted as an intersection of (a) and (b), and shows the region for low influence in the space of independent mechanisms.  }
\label{fig: general to independent}
\end{figure*}

\begin{proposition} \label{pro: joint mechanism LI constraint}
Let $\mathcal{M}: \mathcal{D} \to \mathcal{V}$ such that $P_{v_1,\ldots,v_{|\mathcal{D}|}} = \Pr[\mathcal{M}(d_1)=v_1, \ldots, \mathcal{M}(d_{|\mathcal{D}|})=v_{|\mathcal{D}|}]$. Then $\mathcal{M}$ is $\iota$-LI if and only if, for every $d_i \sim d_k$,
\begin{align*}
    \sum_{ (v_1,\ldots, v_{|\mathcal{D}|}) \in \mathcal{V}^{|\mathcal{D}|} : v_i \neq v_k  } P_{v_1,\ldots,v_{|\mathcal{D}|}} \leq \iota .
\end{align*}
In other words, $\iota$-LI is a linear constraint in the space $\mathfrak{M}$.
\end{proposition}

\begin{proof}
The mechanism $\mathcal{M}$ is $\iota$-LI if and only if, for any two neighboring datasets $d_i \sim d_k$, 
\begin{align*}
    \iota \geq \Pr [\mathcal{M}(d_i) \neq \mathcal{M}(d_k)] = \sum_{ (v_1,\ldots, v_{|\mathcal{D}|}) \in \mathcal{V}^{|\mathcal{D}|} : v_i \neq v_k  } P_{v_1,\ldots,v_{|\mathcal{D}|}} .
\end{align*}

\end{proof}

We revisit the toy example in the introduction and show that $\iota$-LI implies in $(0,\iota)$-DP.

\begin{example}
Consider the case where $\mathcal{D}$ consists of two neighboring datasets ${d_1 \sim d_2}$, and the output space is binary, i.e. $\mathcal{V} = \{1,2 \}$. A general mechanism $\mathcal{M}: \mathcal{D} \rightarrow \mathcal{V}$ is then determined by the matrix $P=\left( \begin{smallmatrix} P_{11} & P_{12} \\ P_{21} & P_{22} \end{smallmatrix} \right)$, where $P_{ij} = \Pr[\mathcal{M}(d_1)=i , \mathcal{M}(d_2)=j]$. We note that the marginals are given by $\Pr[\mathcal{M}(d_1) = k] = P_{k1}+P_{k2}$ and $\Pr[\mathcal{M}(d_2) = k] = P_{1k}+P_{2k}$. Thus, a simple calculation shows that, $\mathcal{M}$ is $(0,\delta)$-DP if and only if $|P_{12}-P_{21}| \leq \delta$, i.e. if the entries in the anti-diagonal of $P$ are close. By Proposition \ref{pro: joint mechanism LI constraint}, the mechanism $\mathcal{M}$ is $\iota$-LI if and only if $P_{12}+P_{21} \leq \iota$, i.e. if the sum of the entries in the anti-diagonal of $P$ is close to zero. Thus, since $|P_{12}-P_{21}| \leq P_{12}+P_{21} \leq \iota$, it follows that if the mechanism $\mathcal{M}$ is $\iota$-LI, it is also $(0,\iota)$-DP.
\end{example}

The statement in the example above holds true in general. We first prove it for value differential privacy and then combine it with Proposition \ref{pro: vdp equiv dp} to prove the general result.

\begin{proposition} \label{pro: li implies vdp}
Let $\mathcal{M}: \mathcal{D} \to \mathcal{V}$ be an $\iota$-LI mechanism. Then, $\mathcal{M}$ is $(0,\iota)$-VDP.
\end{proposition}

\begin{proof}
Let $d_1,d_2 \in \mathcal{D}$ be two neighboring databases and $P$ be the joint probability matrix given by $P_{ij} = \Pr[\mathcal{M}(d_1)=v_i , \mathcal{M}(d_2)=v_j]$. Then, for every $k \in [|\mathcal V|]$,
\begin{align*}
    \iota &\geq \Pr[\mathcal{M}(d_1) \neq \mathcal{M}(d_2)] = \sum_{(i,j):i\neq j} P_{ij} \\
    &\geq \sum_{i \neq k} (P_{ki}+P_{ik}) \geq \sum_{i \neq k} | P_{ki} - P_{ik} | \\
    &\geq \left | \sum_{i \neq k} (P_{ki} - P_{ik}) \right | = \left | \sum_{j} P_{kj} - \sum_{i} P_{ik} \right | .
\end{align*}
But $\sum_{j} P_{kj} = \Pr[\mathcal{M}(d_1)=v_k]$ and $\sum_{i} P_{ik} = \Pr[\mathcal{M}(d_2)=v_k]$. Thus, 
\[ \Pr[\mathcal{M}(d_1)=v_k] \leq \Pr[\mathcal{M}(d_2)=v_k] + \iota \quad \text{and} \quad \Pr[\mathcal{M}(d_2)=v_k] \leq \Pr[\mathcal{M}(d_1)=v_k] + \iota ,\] for every $k \in \mathcal{V}$, i.e., $\mathcal{M}$ is $(0,\iota)$-value differentially private.
\end{proof}

\begin{figure*}[t]
\definecolor{blizzardblue}{rgb}{0.67, 0.9, 0.93}
\definecolor{grannysmithapple}{rgb}{0.66, 0.89, 0.63}
\definecolor{lavenderblue}{rgb}{0.8, 0.8, 1.0}
\begin{subfigure}[]{0.45\textwidth}
    \centering
    \begin{tabular}{@{}llll@{}}
$P_{11}$ & \cellcolor{blizzardblue}$P_{12}$ & $P_{13}$ & $P_{14}$ \\
\cellcolor{grannysmithapple}$P_{21}$ & $P_{22}$ & \cellcolor{grannysmithapple}$P_{23}$ & \cellcolor{grannysmithapple}$P_{24}$ \\
$P_{31}$ & \cellcolor{blizzardblue}$P_{32}$ & $P_{33}$ & $P_{34}$ \\
$P_{41}$ & \cellcolor{blizzardblue}$P_{42}$ & $P_{43}$ & $P_{44}$
\end{tabular}
    \caption{$(0,\delta)$-VDP}
    \label{tab:a example}
\end{subfigure}
\begin{subfigure}[]{0.45\textwidth}
    \centering
    \begin{tabular}{@{}llll@{}}
$P_{11}$ & \cellcolor{lavenderblue}$P_{12}$                         & \cellcolor{lavenderblue}$P_{13}$                         & \cellcolor{lavenderblue}$P_{14}$                         \\
\cellcolor{lavenderblue}$P_{21}$                         & $P_{22}$ & \cellcolor{lavenderblue}$P_{23}$                         & \cellcolor{lavenderblue}$P_{24}$                         \\
\cellcolor{lavenderblue}$P_{31}$                         & \cellcolor{lavenderblue}$P_{32}$                         & $P_{33}$ & \cellcolor{lavenderblue}$P_{34}$                         \\
\cellcolor{lavenderblue}$P_{41}$                         & \cellcolor{lavenderblue}$P_{42}$                         & \cellcolor{lavenderblue}$P_{43}$                         & $P_{44}$
\end{tabular}
    \caption{$\iota$-Low Influence}
    \label{tab:b example}
\end{subfigure}
\caption{The main idea behind the proof of Proposition \ref{pro: li implies vdp}. The elements of the matrix signify $P_{ij} = \Pr[\mathcal{M}(d_1)=v_i , \mathcal{M}(d_2)=v_j]$. In Figure (a), $(0,\delta)$-VDP requires the difference between the sum of row $k$ and column $k$ to be small. In Figure (b), $\iota$-LI requires the sum of all non-diagonal elements to be small. The latter implies the former. However, one can easily construct counterexample to show the reverse is generally not true. }
\label{tab: example}
\end{figure*}

\begin{theorem} \label{teo: li implies dp}
Let $\mathcal{M}: \mathcal{D} \to \mathcal{V}$ be an $\iota$-LI mechanism. Then, $\mathcal{M}$ is $(0,\iota(|\mathcal{V}|-1))$-DP.
\end{theorem}

\begin{proof}
Follows directly from Propositions \ref{pro: vdp equiv dp} and \ref{pro: li implies vdp}.
\end{proof}

The idea behind Proposition \ref{pro: li implies vdp} is shown in Fig.~\ref{tab: example}. The $\iota$-LI conditions require that for any two $d_1 \sim d_2$, the sum of non-diagonal elements of the joint probability matrix  $P_{ij} = \Pr[\mathcal{M}(d_1)=v_i, \mathcal{M}(d_2)=v_j]$ be smaller than $\iota$. On the other hand, $(0,\iota)$-VDP conditions require that for all $k \in [|\mathcal V|]$, the difference between the sum of row $k$ and the sum of column $k$ of the same matrix be small. Clearly, if the sum of all non-diagonal elements is small, the difference between the sum of any two subsets of non-diagonal elements is also small. For value differential privacy, it suffices to consider the difference between the sum of row $k$ and column $k$, noting that the diagonal element $P_{kk}$ (which could be large) will cancel out in the difference.

%\sal{Needs to update this paragraph. What to say about $\epsilon = 0$.} The result in Theorem~\ref{teo: li implies dp} is essentially a negative result.
%Indeed, in most practical $(\epsilon,\delta)$-DP scenarios, the value of $\delta$ should be rather small, i.e. in the order of $10^{-3}$.
%A $\iota$-LI mechanism, with $\iota > 10^{-3}$, would thus result in a mechanism which is not deemed to be sufficiently private.
In most practical $(\epsilon,\delta)$-DP scenarios, the value of $\delta$ should be rather small, e.g., in the order of $10^{-3}$. According to Theorem \ref{pro: li implies vdp}, setting $\iota = \frac{\delta}{|\mathcal V|-1}$ when designing a low influence mechanism will ensure the desired $\delta$.
However, this method is only valid for $(\epsilon,\delta)$-DP with $\delta > 0$. Indeed, if one desires $\delta = 0$, i.e. the classical $\epsilon$-DP condition, a mechanism being low-influence is not sufficient to guarantee privacy unless $\iota = 0$.

\section{Useful Low Influence Differentially Private Schemes} \label{sec: nontrivial lowinfluence dp schemes}

Most problems in differential privacy exhibit a tradeoff between privacy and some notion of utility.
This is typically captured by a utility function $U: \mathfrak{M} \rightarrow \mathbb{R}$, which is a function from the space of randomized mechanisms to the real-values which usually captures how well a randomized mechanism approximates a desired output if no privacy mechanism was to be put in place.
Problems in differential privacy can then be formulated as:  $\underset{ \mathcal{M} \in \mathfrak{M} }{\text{maximize }} U(\mathcal{M})$, subject to differential privacy constraint, and potentially other practical constraints. Obviously, the structure of the optimization problem, its computational complexity, and the overall efficacy of the solution are three important considerations when designing differentially private mechanisms.

As shown in Proposition \ref{pro: linear in I is linear in M}, every linear constraint on the space of independent mechanisms $\mathfrak{I}$ is also linear on the space of joint mechanisms $\mathfrak{M}$. In particular, $(\epsilon, \delta)$-DP and $(\epsilon, \delta)$-VDP are linear constraints over both $\mathfrak{I}$ and $\mathfrak{M}$. However, the opposite is not true. As we have seen, low influence constraints are linear over $\mathfrak{M}$, but non-linear over $\mathfrak{I}$. Linear conditions are desired as efficient linear program solvers can be used to find the solution.

This simplicity in the constraints comes at the price of dimensionality. The degrees of freedom in independent mechanisms is $\mathcal{O} (|\mathcal{D}| \times |\mathcal{V}|)$, whereas the degrees of freedom in general joint mechanisms is $\mathcal{O} ({|\mathcal{V}|}^{|\mathcal{D}|})$.
Clearly, the desired linearity and increase in the degrees of freedom come at the cost of computational complexity. 
Solving the tradeoff between performance and complexity in designing joint mechanisms is an interesting topic of research, but is beyond the scope of this paper. 
However, in the remainder of this section we motivate the potential benefits of considering general joint mechanisms. 
We start by a general result on the existence of nontrivial mechanisms with arbitrarily low influence (which according to Theorem \ref{teo: nontrivial indep not li} must necessarily be non-independent).  

%However, linearity comes at a dimension (complexity) cost, because $\mathfrak{I}$ and $\mathfrak{M}$ live in spaces of dimension $|\mathcal{D}| \times |\mathcal{V}|$ and ${|\mathcal{V}|}^{|\mathcal{D}|}$, respectively. Nevertheless, dealing with only linear objective functions and constraints is advantageous as linear program solvers can much more efficient than . 

\begin{comment}
\begin{equation*}
\begin{aligned}
& \underset{ \mathcal{M} \in \mathfrak{M} }{\text{maximize}}
& & U(\mathcal{M}) \\
& \text{subject to}
& & \text{Constraints}.
\end{aligned}
\end{equation*}
\end{comment}

%By Proposition \ref{pro: linear in I is linear in M},  However, the opposite is not true, as we have seen occur with the low influence constraint. This linearity comes at a dimension (complexity) cost, because $\mathfrak{I}$ and $\mathfrak{M}$ live in spaces of dimension $|\mathcal{D}| \times |\mathcal{V}|$ and ${|\mathcal{V}|}^{|\mathcal{D}|}$, respectively. 

\begin{theorem}
There exists nontrivial mechanisms $\mathcal{M}: \mathcal{D} \to \mathcal{V}$ with arbitrarily low influence.
\end{theorem}

\begin{proof}
Without loss of generality, let $\mathcal{D}=\{ d_1, \ldots, d_{|\mathcal{D}|} \}$ such that $d_1 \sim d_2$, $\mathcal{V} = \{1, \ldots, |\mathcal{V}| \}$, and choose $0<\alpha <1$. Consider the mechanism $P_{v_1 , \ldots , v_{|\mathcal{D}|}} \in \mathfrak{M}$ such that $P_{1,\ldots,1}=P_{2,\ldots,2}=\frac{1-\alpha}{2}$, $P_{1,2,1,\ldots,1} = \alpha$, and $P_{v_1 , \ldots , v_{|\mathcal{D}|}}=0$ for the other $(v_1 , \ldots , v_{|\mathcal{D}|}) \in \mathcal{V}^{|\mathcal{D}|}$. Then, 
\begin{align*}
    \Pr[\mathcal{M}(d_1)=1] = \sum_{(v_1 , \ldots , v_{|\mathcal{D}|}) \in \mathcal{V}^{|\mathcal{D}|} : v_1 = 1} P_{v_1 , \ldots , v_{|\mathcal{D}|}} = P_{1,\ldots,1} + P_{1,2,1,\ldots,1} = \frac{1}{2} + \frac{\alpha}{2} 
\end{align*}
and
\begin{align*}
    \Pr[\mathcal{M}(d_2)=2] = \sum_{(v_1 , \ldots , v_{|\mathcal{D}|}) \in \mathcal{V}^{|\mathcal{D}|} : v_2 = 2} P_{v_1 , \ldots , v_{|\mathcal{D}|}} = P_{2,\ldots,2} + P_{1,2,1,\ldots,1} = \frac{1}{2} + \frac{\alpha}{2} .
\end{align*}
Since both values are larger than $1/2$, it follows that 
\[ \argmax_{v \in \mathcal{V}} \Pr[\mathcal{M}(d_1) = v] = 1 \neq 2 = \argmax_{v \in \mathcal{V}} \Pr[\mathcal{M}(d_2) = v] .\]
And thus, $\mathcal{M}$ is nontrivial.

We now show that $\mathcal{M}$ is $\alpha$-LI. This follows directly from the fact that
\begin{align*}
    \sum_{(v_1 , \ldots , v_{|\mathcal{D}|}) \in \mathcal{V}^{|\mathcal{D}|} : v_i \neq v_k} P_{v_1 , \ldots , v_{|\mathcal{D}|}} = \left\{\begin{matrix}
P_{1,2,1,\ldots,1} = \alpha & \text{if $v_i = 2$ or $v_k = 2$,} \\ 
0 & \text{otherwise.}
\end{matrix}\right.
\end{align*}
\end{proof}

% In the following, we illustrate how considering joint  versus independent mechanisms can impact the optimization. \rev{First, we provide a definition for balanced mechanisms.}

% \rev{\begin{definition}
% A mechanism $\mathcal{M}: \mathcal{D} \to \mathcal{V}$ is called balanced if $M^*(d) \triangleq \max_{v\in \mathcal{V}} \Pr[\mathcal{M}(d)=v]$ is a constant function, independent of the dataset $d$.
% \end{definition}}
% \rev{Roughly speaking, a well designed nontrivial and balanced mechanism should show variability in terms of $\argmax_v \Pr[\mathcal{M}(d) = v]$, but non-variability in terms of the maximum value. More strongly, the most likely output in a well designed mechanism should be the true query output and the likelihood should be balanced across different realizations of dataset. }
\begin{example}\label{ex:ind_mechanism}

Let $\mathcal{D}$ consist of two neighboring datasets $d_1 \sim d_2$ and the output set be the binary set $\mathcal{V} = \{ 1,2 \}$. We wish to design a mechanism which does not favor one database over another. This can be achieved through imposing $\Pr[\mathcal{M}(d_1)=1] = \Pr[\mathcal{M}(d_2)=2]$, which we call the balancing constraint. Consider the utility function given by $U(\mathcal{M}) = \Pr[\mathcal{M}(d_1)=1]$.  Restricting ourselves to independent mechanisms we consider the following optimization problem.

\begin{equation*}
\begin{aligned}
& \underset{ \mathcal{M} \in \mathfrak{I} }{\text{maximize}}
& & U(\mathcal{M}) \\
& \text{subject to}
& & \text{$\mathcal{M}$ satisfies $(\epsilon,\delta)$-DP and} \\
& & & \text{$\Pr[\mathcal{M}(d_1)=1] = \Pr[\mathcal{M}(d_2)=2]$} .
\end{aligned}
\end{equation*}

Using the characterization of $\mathfrak{I}$ given in Proposition \ref{def: independent mechanism}, and setting $x = \Pr[\mathcal{M}(d_1)=1] = \Pr[\mathcal{M}(d_2)=2]$, the optimization problem is as follows.

\begin{equation*}
\begin{aligned}
& \underset{ x \in \mathbb{R} }{\text{maximize}}
& & x \\
& \text{subject to}
& & x \leq e^\epsilon (1-x) + \delta \\
& & & (1-x) \leq e^\epsilon x + \delta \\
& & & 0 \leq x \leq 1 .
\end{aligned}
\end{equation*}

The solution to this optimization problem is $x^* = \frac{e^{\epsilon} + \delta}{e^{\epsilon}+1} \geq \frac{1}{2}$, confirming that the mechanism is nontrivial. As we previously saw, the influence of this mechanism is lower bounded by $\frac{1}{2}$. Indeed, the influence is given by 
\[ I(x^*) = 1 - \frac{2(1-\delta)(e^{\epsilon}+\delta)}{(e^{\epsilon}+1)^2} \geq \frac{1}{2} .\]

\end{example}

We now show that if we consider general mechanisms, we can lower the influence. 

\begin{example}\label{ex:joint_mechanism}

Let $\mathcal{D}$ consist of two neighboring datasets $d_1 \sim d_2$ and the output set be the binary set $\mathcal{V} = \{ 1,2 \}$. We consider the utility function given by $U(\mathcal{M}) = \Pr[\mathcal{M}(d_1)=1]$ with the balancing constraint that $\Pr[\mathcal{M}(d_1)=1] = \Pr[\mathcal{M}(d_2)=2]$ as in Example \ref{ex:ind_mechanism}. We now consider general mechanisms.

\begin{equation*}
\begin{aligned}
& \underset{ \mathcal{M} \in \mathfrak{M} }{\text{maximize}}
& & U(\mathcal{M}) \\
& \text{subject to}
& & \text{$\mathcal{M}$ satisfies $(\epsilon,\delta)$-DP and} \\
& & & \text{$\Pr[\mathcal{M}(d_1)=1] = \Pr[\mathcal{M}(d_2)=2]$} .
\end{aligned}
\end{equation*}
Using the characterization of $\mathfrak{M}$ given in Proposition \ref{def: joint mechanism}, and setting $x = \Pr[\mathcal{M}(d_1)=1 , \mathcal{M}(d_2)=2]$ and $y = \Pr[\mathcal{M}(d_1)=2 , \mathcal{M}(d_2)=1]$, the optimization problem is as follows.
\begin{equation*}
\begin{aligned}
& \underset{ (x,y) \in \mathbb{R}^2 }{\text{maximize}}
& & \frac{1+x-y}{2} \\
& \text{subject to}
& & \frac{1+x-y}{2} \leq e^\epsilon \frac{1+y-x}{2} + \delta \\
& & & \frac{1+y-x}{2} \leq e^\epsilon \frac{1+x-y}{2} + \delta \\
& & & 0 \leq x \leq 1 \\
& & & 0 \leq y \leq 1 \\
& & & 0 \leq x+y \leq 1 .
\end{aligned}
\end{equation*}

If we also minimize the influence, given by $I(x,y)=x+y$, we get the solution $x^* = \frac{e^\epsilon +2\delta -1}{e^\epsilon +1}$ and $y^*=0$ with influence
\[ I(x^*,y^*)= \frac{e^\epsilon +2\delta -1}{e^\epsilon +1} \leq 1 - \frac{2(1-\delta)(e^\epsilon+\delta)}{(e^\epsilon+1)^2} .\] It can be verified that the above $x^*$ and $y^*$ ensure nontrivial mechanisms for all $\epsilon \geq 0$ and $0 \leq \delta \leq 1$. Note that for $\epsilon=0$ and $\delta=0$, $I(x,y) = 0$. Thus, in  this case, the influence is not lower bounded.
\end{example}

\begin{figure}
    \centering
    \begin{tikzpicture}
\pgfplotsset{compat=1.14}
\definecolor{ao}{rgb}{0.13, 0.55, 0.13}
% let both axes use the same layers
\pgfplotsset{set layers}
\begin{axis}[width=5cm, height=5cm,thick,smooth,no markers,xmin=0.5,xmax=1,ymin=0,ymax=1,axis on top,xtick={0,0.5,0.75,1},ytick={0,0.5,1},ylabel near ticks,xlabel near ticks,samples=200,
scale only axis,
axis y line*=left,
xlabel={Utility $U(\mathcal{M})$},
ylabel style = {align=center},
ylabel={Privacy level $\epsilon$ \ref{pgfplots:plot1}},
]
\addplot[blue,ultra thick] {ln(x/(1-x)};
\label{pgfplots:plot1}
\end{axis}
\begin{axis}[width=5cm, height=5cm,thick,smooth,no markers,xmin=0.5,xmax=1,ymin=0,ymax=1,axis on top,xtick={0,0.5,0.75,1},ytick={0,0.5,1},xlabel near ticks,samples=200,
scale only axis,
axis y line*=right,
axis x line=none,
ylabel style = {align=center},
ylabel={Independent - Influence\ref{pgfplots:plot2} \\ Joint - Influence\ref{pgfplots:plot3}},
]
\addplot [red, ultra thick] {2*x*x-2*x+1};
\label{pgfplots:plot2}
\addplot [color=ao, , ultra thick] {2*x-1};
\label{pgfplots:plot3}
\end{axis}
\end{tikzpicture}
    \caption{The privacy-utility-influence tradeoff for joint and independent mechanisms. While independent mechanisms are sufficient to achieve the best privacy-utility tradeoff, they fall short in terms of influence -- for any utility level, there is a joint scheme which achieves the same privacy and a lower influence.}
    \label{fig:tradeoffs}
\end{figure}
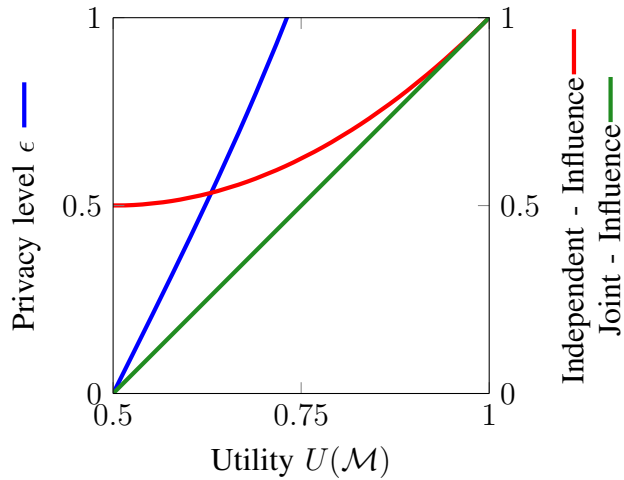

In Figure \ref{fig:tradeoffs} we show the optimal tradeoffs between $(\epsilon,0)$-DP, the utility $U(\mathcal{M})$, and the influence $\iota$ of joint, and independent randomized mechanisms. Note that this tradeoff can be obtained by letting $\delta = 0$ in the solutions in Examples~\ref{ex:ind_mechanism} and \ref{ex:joint_mechanism}, and by then varying the value of $\epsilon$.
As can be seen in the figure, both independent and joint mechanisms exhibit the same privacy-utility tradeoff, i.e. for any $\epsilon$, there is an $(\epsilon,0)$-DP independent mechanism that achieves the same utility as the best $(\epsilon,0)$-DP joint mechanism.
On the other hand, joint mechanisms achieve a strictly better optimal tradeoff in terms of their influence $\iota$.
That is, while independent mechanisms are sufficient to achieve the optimal privacy (captured by the parameter $\epsilon$ in the $(\epsilon,0)$-DP guarantee) and utility tradeoff, they fall short at providing the best influence level $\iota$.

\section{Concluding Remarks}
The study in this paper reveals the interesting interactions between differential privacy and utility, on the one hand, and mechanism influence and independence on the other. In the independent differentially private mechanisms prevalent in the literature, the additive noise is independent across input databases. Such independence enables relatively easy design of the mechanism in terms of the privacy-utility tradeoff. While such mechanisms are casually understood to offer ``low influence" of an individual on the query output, we show that they are not necessarily statistically low influence in a formal sense. Even further, we show that if one were to enforce an arbitrarily low influence on an independent differentially private mechanism, a heavy price would have to be paid: the mechanism becomes trivial (statistically-speaking constant) and hence loses any practical utility. The more general class of joint mechanisms not only allows differential privacy, low influence, and utility to be jointly achieved, they also make the optimization of influence linear over the mechanism space. However, this increases the dimension of the mechanism space over which such optimization takes place.

There are several interesting directions for future research. In \cite{7498650}, the authors studied the relation between individual \emph{identifiability} and differential privacy. Roughly speaking, identifiability aims to measures the guessing capability of an adversary who knows all  elements of the dataset $d$ that are common with its neighboring datasets and wants to guess the remaining unknown element of $d$ after observing the mechanism output $\mathcal{M}(d)$.  Identifiability level, denoted by $\epsilon_d$, of a randomized mechanism depends on the input distribution of elements in $d$ and the $\epsilon$ level in differential privacy. In \cite{cuff:2016}, the authors studied the \emph{worst-case} mutual information (MI) leakage between the $i$-th entry of dataset $d$ and the randomized mechanism output $\mathcal{M}(d)$, given the knowledge of dataset entries excluding $i$ and when maximized over all possible distributions of dataset entries. They proved that such defined $\epsilon_{m}$-MI privacy level is weaker than $(\epsilon,0)$-DP, but does imply $(\epsilon, \sqrt{2\epsilon_m})$-DP for any $\epsilon, \epsilon_m \geq 0$. Exploring the relation between low influence, $\epsilon_d$-identifiability and $\epsilon_m$-MI privacy is an interesting research direction.

Another interesting direction for future research will be systematic low-complexity design of non-independent, high-utility, low-influence, and differentially-private mechanisms. In addition, it is worthwhile to formalize the counterpart of our results for continuous-valued queries and mechanisms. Finally, we remark that the general framework of input-dependent joint mechanisms that we studied in this paper may be of value in time-varying systems where data changes over time according to an inherent dynamical model or is gradually acquired and released (e.g., in social networks or health systems), see for example \cite{dynamical,Koufogiannis_Han_Pappas_2017}. We propose further investigation of a formal connection as future research.

\appendix

In this Appendix we prove Propositions 1, 2, and  Theorem 1 of the main text. 

\begin{propositionapp} 
It holds that,
\begin{multline*}
    \mathfrak{I} \cong \Bigg \{ P_{v_1,\ldots,v_{|\mathcal{D}|}} \in \mathfrak{M} : P_{v_1,\ldots,v_{|\mathcal{D}|}} = \prod_{i=1}^{|\mathcal{D}|} \sum_{({w_1,\ldots,w_{|\mathcal{D}|}}) \in \mathcal{V}^{|\mathcal{D}|} : w_i = v_i }  P_{w_1,\ldots,w_{|\mathcal{D}|}} , \\ \text{for every $(v_1,\ldots,v_{|\mathcal{D}|}) \in \mathcal{V}^{|\mathcal{D}|}$} \Bigg \} .
\end{multline*}
In other words, $\mathfrak{I}$ is a polynomial of degree $|\mathcal{D}|$ in $\mathbb{R}^{{|\mathcal{V}|}^{|\mathcal{D}|}}$.
\end{propositionapp}

\begin{proof}
Let $P \in \mathfrak{M}$ where $P_{v_1, \ldots, v_{|\mathcal{D}|}} = \Pr[\mathcal{M}(d_1)=v_1, \ldots \mathcal{M}(d_{|\mathcal{D}|})=v_{|\mathcal{D}|}]$. Then, 
\[ \Pr[\mathcal{M}(d_k)=v] = \sum_{({w_1,\ldots,w_{|\mathcal{D}|}}) \in \mathcal{V}^{|\mathcal{D}|} : w_k = v }  P_{w_1,\ldots,w_{|\mathcal{D}|}} .\]
Thus, the mechanism is independent if and only if
\[ P_{v_1,\ldots,v_{|\mathcal{D}|}} = \prod_{i=1}^{|\mathcal{D}|} \sum_{({w_1,\ldots,w_{|\mathcal{D}|}}) \in \mathcal{V}^{|\mathcal{D}|} : w_i = v_i }  P_{w_1,\ldots,w_{|\mathcal{D}|}} , \\ \text{for every $(v_1,\ldots,v_{|\mathcal{D}|}) \in \mathcal{V}^{|\mathcal{D}|}$} .\]
\end{proof}

\begin{propositionapp} \label{pro: mech prob}
A randomized mechanism is equivalent to a probability distribution on the space of all functions $h: \mathcal{D} \rightarrow \mathcal{V}$.
\end{propositionapp}

\begin{proof}
The space of functions from $h: \mathcal{D} \rightarrow \mathcal{V}$ is isomorphic to the space $\mathcal{V}^{|\mathcal{D}|}$. Indeed, if we order $\mathcal{D}=\{d_1,\ldots,d_{|\mathcal{D}|} \}$, then each vector $(v_1, \ldots, v_{|\mathcal{D}|}) \in \mathcal{V}^{|\mathcal{D}|}$ corresponds to the function $h: \mathcal{D} \rightarrow \mathcal{V}$ such that $h(d_1) = v_1 , \ldots, h(d_{|\mathcal{D}|})= v_{|\mathcal{D}|}$. By Equation 1 in the main text, every $P \in \mathfrak{M}$ is a probability distribution on $\mathcal{V}^{|\mathcal{D}|}$. Thus, using the above equivalence, every mechanism $P \in \mathfrak{M}$ is a probability distribution on the space of functions from $h: \mathcal{D} \rightarrow \mathcal{V}$.
\end{proof}

We now prove Theorem 1. To do this, we present some preliminary definitions and lemmas. First, recall Definition 9 of the main text $I(u,v) = 1 - \sum_{i=1}^{\mathcal{V}} u_i v_i$, which can be used to measure the local influence of any two datasets. We relate this function to $\iota$-LI for independent mechanisms in the following lemma.

% \begin{definition} \label{def: influence function}
% Let $\mathcal{V}$ be a finite set. We define
% $I: \mathbb{R}^{|\mathcal{V}| \times |\mathcal{V}|} \rightarrow \mathbb{R}$, such that $I(u,v) = 1 - \sum_{i=1}^{\mathcal{V}} u_i v_i$.
% \end{definition}

\begin{lemma} \label{lem: relate I and iota}
Let $\mathcal{M}:\mathcal{D}\rightarrow\mathcal{V}$ be an independent mechanism and denote $M_{ij}=\Pr[\mathcal{M}(d_i)=v_j]$. Then, $\mathcal{M}$ is $\iota$-LI if and only if $I(M_i,M_k) \leq \iota$ for every $d_i \sim d_k$.
\end{lemma}

\begin{proof}
This is a direct restatement of Proposition 5 in the main text.
\end{proof}

It follows from the lemma that a lower bound on $I$ for two neighboring datasets gives us a lower bound on $\iota$.

We now define generalized versions of the sets $R^*$ and $S$ in Example 1 of the main text.

\begin{definition} \label{def: R}
Let $\mathcal{V}$ be a finite set. We define
\begin{multline*}
     R = \bigg \{ (u,v) \in \mathbb{R}^{|\mathcal{V}| \times |\mathcal{V}|} : \max_{i \in [|\mathcal{V}|]} u_i = u_1 , \max_{i \in [|\mathcal{V}|]} v_i = v_2 , \sum_{i=1}^{|\mathcal{V}|} u_i = \sum_{i=1}^{|\mathcal{V}|} v_i = 1 , \\ 0 \leq u_i \leq 1 , 0 \leq v_i \leq 1 , \quad \text{for every $i \in [|\mathcal{V}|]$} \bigg \} .
\end{multline*}
\end{definition}

This set corresponds to nontrivial independent mechanisms in the following way.

\begin{lemma} \label{lem: nontrivial mechanisms and R}
Let $\mathcal{M}:\mathcal{D}\rightarrow\mathcal{V}$ be a nontrivial independent mechanism and $M_{ij} = \Pr[M(d_i)=v_j]$. Then, there exist two neighboring datasets $d_i \sim d_j$ such that $(M_i,M_j) \in R$.
\end{lemma}

\begin{proof}
Since $\mathcal{M}$ is nontrivial, there exist two neighboring datasets $d_i \sim d_j$ such that $$\argmax_{v \in \mathcal{V}} \Pr[M(d_{i})=v] \neq \argmax_{v \in \mathcal{V}} \Pr[M(d_{j})=v].$$ We order $\mathcal{V} = \{ v_1 , \ldots, v_{|V|} \}$ such that $$v_1 = \argmax_{v \in \mathcal{V}} \Pr[M(d_i)=v], \quad v_2 = \argmax_{v \in \mathcal{V}} \Pr[M(d_j)=v].$$ Then, by construction, $(M_i,M_j) \in R$.
\end{proof}

Thus, a lower bound on $I$ subject to $R$ would give us a lower bound on $\iota$ for all nontrivial independent mechanisms. To do this, we will consider the following set.

\begin{definition}
Let $R$ be as in Definition \ref{def: R}. We define
\begin{align*}
    S = \{ (u,v) \in R : u_i=v_i=0 \quad \text{for every $i \geq 3$} \} .
\end{align*}
\end{definition}

We show in Lemma $\ref{lem: I(u^*,v^*) < I(u,v)}$ that the minimum of $I$ over $R$ is in $S$. To do this we need the following construction.

\begin{definition}
Let $(u,v)\in R$, $S_u = \sum_{i=3}^{|\mathcal{V}|} u_i$, $S_v = \sum_{i=3}^{|\mathcal{V}|} v_i$. We define $u^* \in \mathbb{R}^{|\mathcal{V}|}$ such that $u^*_1 = u_1 + \frac{S_u}{2}$, $u^*_2 = u_2 + \frac{S_u}{2}$, and $u^*_i = 0$ for $i\in |\mathcal{V}|-\{1,2\}$. We also define $v^* \in \mathbb{R}^{|\mathcal{V}|}$ such that $v^*_1 = v_1 + \frac{S_v}{2}$, $v^*_2 = v_2 + \frac{S_v}{2}$, and $v^*_i = 0$ for $i\in |\mathcal{V}|-\{1,2\}$. 
\end{definition}

Thus, for any $(u,v)\in R$ we can construct the point $(u^*,v^*)$.

\begin{lemma}
Let $(u,v)\in R$. Then, $(u^*,v^*) \in S \subseteq R$.
\end{lemma}

\begin{proof}
Since $u_1 \geq u_2$, it follows that 
\[ u^*_1 = u_1 + \frac{S_u}{2} \geq u_2 + \frac{S_u}{2} = u^*_2 \geq u_i = 0 \quad \text{for every $i \in |\mathcal{V}|-\{1,2\}$}. \]
Also, $\sum_{i=1}^{|\mathcal{V}|} u_i = u_1 + \frac{S_u}{2} + u_2 + \frac{S_u}{2} =1$. The same arguments hold for $v^*$. Thus, $(u^*,v^*) \in \mathbb{R}^{|\mathcal{V}| \times |\mathcal{V}|}$.
\end{proof}

We now show that the minimum of $I$ over $R$ is in $S$.

\begin{lemma} \label{lem: I(u^*,v^*) < I(u,v)}
Let $(u,v)\in R$. Then, $I(u^*,v^*) \leq I(u,v)$. And therefore, 
\[ \min_{(u,v) \in R} I(u,v) = \min_{(u,v) \in S} I(u,v) .\]
\end{lemma}

\begin{proof}
Since,
\begin{align*}
    \sum_{i=3}^{|\mathcal{V}|} u_1 v_i \geq \sum_{i=3}^{|\mathcal{V}|} u_i v_i \quad \text{and} \quad \sum_{i=3}^{|\mathcal{V}|} u_i v_2 \geq \sum_{i=3}^{|\mathcal{V}|} u_i v_i ,
\end{align*}
it follows that,
\begin{align*}
    I(u^*,v^*) &= 1 - \sum_{i=1}^{|\mathcal{V}|} u^*_i v^*_i = 1 - \left ( u_1 + \frac{S_u}{2} \right ) \left ( v_1 + \frac{S_v}{2} \right ) - \left ( u_2 + \frac{S_u}{2} \right ) \left ( v_2 + \frac{S_v}{2} \right ) \\
    &= 1- u_1 v_1 - u_1 \frac{S_v}{2} - v_1 \frac{S_u}{2} - \frac{S_u}{2} \frac{S_v}{2} - u_2 v_2 - u_2 \frac{S_v}{2} - v_2 \frac{S_u}{2} - \frac{S_u}{2} \frac{S_v}{2} \\
    &\leq 1 - u_1 v_1 - u_2 v_2 -  u_1 \frac{S_v}{2} - v_2 \frac{S_u}{2} = 1 - u_1 v_1 - u_2 v_2 - \frac{1}{2} \sum_{i=3}^{|\mathcal{V}|} u_1 v_i - \frac{1}{2} \sum_{i=3}^{|\mathcal{V}|} u_i v_2 \\
    &\leq 1 - \sum_{i=1}^{|\mathcal{V}|} u_i v_i = I(u,v).
\end{align*}
\end{proof}

To conclude our proof we must minimize $I$ over $S$.

\begin{lemma} \label{lem: min I in S}
Let $(u,v) \in S$. Then, $\min_{(u,v) \in S} I(u,v) = \frac{1}{2}$
\end{lemma}

\begin{proof}
Since $u^*_2 = 1-u^*_1$, $v^*_2 = 1-v^*_1$, minimizing $I$ over $R$ is equivalent to minimizing over the same two-dimensional space defined in (2) of Example 1 in the main text,
\[R^* = \left \{(u,v)\in \mathbb{R}^2 : \frac{1}{2} \leq u \leq 1 \quad \text{and} \quad 0 \leq v \leq \frac{1}{2} \right \} .\]
Thus, as shown in the example, $\min_{(u,v) \in R^*} I(u,v) = \frac{1}{2}$.
\end{proof}

We now prove Theorem 1 of the main text.

\begin{theoremapp} 
Let $\mathcal{M}: \mathcal{D} \to \mathcal{V}$ be an independent, $\iota$-LI, nontrivial mechanism. Then, $\iota \geq \frac{1}{2}$. Furthermore, this bound is tight.
\end{theoremapp}

\begin{proof}
By Lemma \ref{lem: nontrivial mechanisms and R}, there exists two neighboring datasets $d_i \sim d_k$ such that $(M_i,M_j) \in R$. Lemmas \ref{lem: I(u^*,v^*) < I(u,v)} and \ref{lem: min I in S} imply that $I(M_i,M_j) \geq \frac{1}{2}$. But then, by Lemma \ref{lem: relate I and iota}, $\iota \geq \frac{1}{2}$. 

To see that the bound is tight, assume $d_1 \sim d_2$ and set $M_{11}=M_{12}=\frac{1}{2}$, $M_{1j}=0$ for $j>2$, $M_{i2}=1$, and $M_{ij}=0$ for all the other $i$ and $j$. Then, 
\[ I(M_i,M_k)= \left\{\begin{matrix}
\frac{1}{2} & \text{if $i=1$ or $k=1$ ,}\\ 
0 & \text{otherwise}  .
\end{matrix}\right. \]
\end{proof}

\section*{Acknowledgement}

The authors would like to thank the reviewers for their valuable comments contributing to the improvement of the paper. The work of P. Sadeghi was supported by the ARC Future Fellowship, FT190100429.

\bibliographystyle{IEEEtran}
\bibliography{references.bib}

\begin{thebibliography}{10}
\providecommand{\url}[1]{#1}
\csname url@samestyle\endcsname
\providecommand{\newblock}{\relax}
\providecommand{\bibinfo}[2]{#2}
\providecommand{\BIBentrySTDinterwordspacing}{\spaceskip=0pt\relax}
\providecommand{\BIBentryALTinterwordstretchfactor}{4}
\providecommand{\BIBentryALTinterwordspacing}{\spaceskip=\fontdimen2\font plus
\BIBentryALTinterwordstretchfactor\fontdimen3\font minus
  \fontdimen4\font\relax}
\providecommand{\BIBforeignlanguage}[2]{{%
\expandafter\ifx\csname l@#1\endcsname\relax
\typeout{** WARNING: IEEEtran.bst: No hyphenation pattern has been}%
\typeout{** loaded for the language `#1'. Using the pattern for}%
\typeout{** the default language instead.}%
\else
\language=\csname l@#1\endcsname
\fi
#2}}
\providecommand{\BIBdecl}{\relax}
\BIBdecl

\bibitem{Differential_privacy}
C.~Dwork, F.~McSherry, K.~Nissim, and A.~Smith, ``Calibrating noise to
  sensitivity in private data analysis,'' in \emph{Theory of Cryptography},
  S.~Halevi and T.~Rabin, Eds.\hskip 1em plus 0.5em minus 0.4em\relax Berlin,
  Heidelberg: Springer Berlin Heidelberg, 2006, pp. 265--284.

\bibitem{dwork2008differential}
C.~Dwork, ``Differential privacy: A survey of results,'' in \emph{International
  Conference on Theory and Applications of Models of Computation}.\hskip 1em
  plus 0.5em minus 0.4em\relax Springer, 2008, pp. 1--19.

\bibitem{dwork2014algorithmic}
C.~Dwork, A.~Roth \emph{et~al.}, ``The algorithmic foundations of differential
  privacy,'' \emph{Foundations and Trends{\textregistered} in Theoretical
  Computer Science}, vol.~9, no. 3--4, pp. 211--407, 2014.

\bibitem{apple}
\BIBentryALTinterwordspacing
D.~P. Team, \emph{Learning with Privacy at Scale}, 2017 (last accessed May
  2020). [Online]. Available:
  \url{https://machinelearning.apple.com/2017/12/06/learning-with-privacy-at-scale.html}
\BIBentrySTDinterwordspacing

\bibitem{uscensus}
\BIBentryALTinterwordspacing
\emph{Disclosure Avoidance and the 2020 Census}, 2020 (last accessed May 2020).
  [Online]. Available:
  \url{https://www.census.gov/about/policies/privacy/statistical_safeguards/disclosure-avoidance-2020-census.html}
\BIBentrySTDinterwordspacing

\bibitem{dwork:rothblum:vadhan;focs:2010}
C.~Dwork, G.~N. Rothblum, and S.~Vadhan, ``Boosting and differential privacy,''
  in \emph{Annual Symposium on Foundations of Computer Science}, 2010, pp.
  51--60.

\bibitem{kifer;ccs:2018}
Z.~Ding, Y.~Wang, G.~Wang, D.~Zhang, and D.~Kifer, ``Detecting violations of
  differential privacy,'' in \emph{ACM SIGSAC Conference on Computer and
  Communications}, 2018, pp. 475--489.

\bibitem{dwork:mironov:delta:2006}
C.~Dwork, K.~Kenthapadi, F.~McSherry, I.~Mironov, and N.~Moni, ``Our data,
  ourselves: Privacy via distributed noise generation,'' in \emph{Annual
  International Conference on the Theory and Applications of Cryptographic
  Techniques}.\hskip 1em plus 0.5em minus 0.4em\relax Springer, 2006, pp.
  486--503.

\bibitem{penrose1946elementary}
L.~S. Penrose, ``The elementary statistics of majority voting,'' \emph{Journal
  of the Royal Statistical Society}, vol. 109, no.~1, pp. 53--57, 1946.

\bibitem{o2014analysis}
R.~O'Donnell, \emph{Analysis of Boolean functions}.\hskip 1em plus 0.5em minus
  0.4em\relax Cambridge University Press, 2014.

\bibitem{nissim:smith:query:dependent:2007}
K.~Nissim, S.~Raskhodnikova, and A.~Smith, ``Smooth sensitivity and sampling in
  private data analysis,'' in \emph{ACM Symp. Theory Comput. (STOC)}, 2007, p.
  75–84.

\bibitem{Viswanath_DP_Staircase}
Q.~{Geng} and P.~{Viswanath}, ``The optimal noise-adding mechanism in
  differential privacy,'' \emph{IEEE Trans. on Inf. Theory}, vol.~62, no.~2,
  pp. 925--951, 2016.

\bibitem{abs:lens}
\BIBentryALTinterwordspacing
J.~Bailie and C.-H. Chien, ``{ABS} perturbation methodology through the lens of
  differential privacy,'' 2019. [Online]. Available:
  \url{http://www.unece.org/fileadmin/DAM/stats/documents/ece/ces/ge.46/2019/mtg1/SDC2019\_S2\_ABS\_Bailie\_D.pdf}
\BIBentrySTDinterwordspacing

\bibitem{gm:paper}
\BIBentryALTinterwordspacing
A.~Ghosh, T.~Roughgarden, and M.~Sundararajan, ``Universally utility-
  maximizing privacy mechanisms,'' in \emph{Proceedings of the forty-first
  annual ACM symposium on Theory of computing {(STOC)}}, 2009, p. 351–360.
  [Online]. Available: \url{http://doi.acm.org/10.1145/1536414.1536464}
\BIBentrySTDinterwordspacing

\bibitem{dp:finite:computers:arxiv}
Q.~{Geng} and P.~{Viswanath}, ``Differential privacy on finite computers,''
  \emph{Journal of Privacy and Confidentiality}, vol.~9, no.~2, pp. 1--46,
  2019.

\bibitem{Geometric_Minimax}
M.~Gupte and M.~Sundararajan, ``Universally optimal privacy mechanisms for
  minimax agents,'' in \emph{ACM SIGMOD-SIGACT-SIGART Symposium on Principles
  of Database Systems}, NY, USA, 2010, p. 135–146.

\bibitem{count:dp:gaussian:arxiv}
\BIBentryALTinterwordspacing
C.~Canonne, G.~Kamath, and T.~Steinke. (2020) The discrete gaussian for
  differential privacy. [Online]. Available:
  \url{https://arxiv.org/abs/2004.00010v2}
\BIBentrySTDinterwordspacing

\bibitem{nytimes}
\BIBentryALTinterwordspacing
G.~Wezerek and D.~Van~Riper, \emph{Changes to the Census Could Make Small Towns
  Disappear}, 2020 (last accessed May 2020). [Online]. Available:
  \url{https://www.nytimes.com/interactive/2020/02/06/opinion/census-algorithm-privacy.html}
\BIBentrySTDinterwordspacing

\bibitem{count:dp;arxiv}
\BIBentryALTinterwordspacing
G.~Cormode, T.~Kulkarni, and D.~Srivastava. (2017) Constrained differential
  privacy for count data. [Online]. Available:
  \url{https://arxiv.org/abs/1710.00608}
\BIBentrySTDinterwordspacing

\bibitem{finite:counting}
\BIBentryALTinterwordspacing
P.~Sadeghi, S.~Asoodeh, and F.~du~Pin~Calmon. (2020) Differentially private
  mechanisms for count queries. [Online]. Available:
  \url{https://arXiv:2007.09374v1}
\BIBentrySTDinterwordspacing

\bibitem{dwork2015preserving}
C.~Dwork, V.~Feldman, M.~Hardt, T.~Pitassi, O.~Reingold, and A.~L. Roth,
  ``Preserving statistical validity in adaptive data analysis,'' in
  \emph{Proceedings of the forty-seventh annual ACM symposium on Theory of
  computing}, 2015, pp. 117--126.

\bibitem{7498650}
W.~{Wang}, L.~{Ying}, and J.~{Zhang}, ``On the relation between
  identifiability, differential privacy, and mutual-information privacy,''
  \emph{IEEE Transactions on Information Theory}, vol.~62, no.~9, pp.
  5018--5029, 2016.

\bibitem{cuff:2016}
P.~Cuff and L.~Yu, ``Differential privacy as a mutual information constraint,''
  in \emph{CSS}, 2016, pp. 43--54.

\bibitem{dynamical}
F.~{Koufogiannis} and G.~J. {Pappas}, ``Differential privacy for dynamical
  sensitive data,'' in \emph{Conference on Decision and Control (CDC)}, 2017,
  pp. 1118--1125.

\bibitem{Koufogiannis_Han_Pappas_2017}
\BIBentryALTinterwordspacing
F.~Koufogiannis, S.~Han, and G.~J. Pappas, ``Gradual release of sensitive data
  under differential privacy,'' \emph{Journal of Privacy and Confidentiality},
  vol.~7, no.~2, Jan. 2017. [Online]. Available:
  \url{https://journalprivacyconfidentiality.org/index.php/jpc/article/view/649}
\BIBentrySTDinterwordspacing

\end{thebibliography}

\end{document}